%% file: main.tex
\documentclass[sigconf]{acmart}

\renewcommand\footnotetextcopyrightpermission[1]{} 
\AtBeginDocument{%
  \providecommand\BibTeX{{%
    \normalfont B\kern-0.5em{\scshape i\kern-0.25em b}\kern-0.8em\TeX}}}

\input{macros}

\begin{document}

\fancyhead{}

\title{Join Size Bounds using $\ell_p$-Norms on Degree Sequences}

\author{Mahmoud Abo Khamis}
\affiliation{%
    \institution{RelationalAI}
    \country{United States}
}

\author{Vasileios Nakos}
\affiliation{%
    \institution{RelationalAI \& University of Athens}
    \country{Greece}
}

\author{Dan Olteanu}
\affiliation{%
    \institution{University of Zurich}
    \country{Switzerland}
}

\author{Dan Suciu}
\affiliation{%
    \institution{University of Washington}
    \country{United States}
}

\renewcommand{\shortauthors}{Abo Khamis et al.}


\input{abstract}



\keywords{query output cardinality, degree sequence, worst-case optimal join}

\maketitle

\input{intro}
\input{comparison}

\input{algorithm}

\input{prelim}

\input{main-proof}
\input{lower-bound}

\input{simple-degrees}

\input{conclusions}

\begin{acks}
    The authors would like to acknowledge Luis Torrejón Machado for their help with the preliminary experiments reported in Appendix~\ref{app:comparison} of this paper.
    
    This work was partially supported by NSF-BSF 2109922, NSF-IIS 2314527 and NSF-SHF 2312195.
    Part of this work was conducted while some of the authors participated in the Simons Program on Logic and Algorithms in Databases and AI.
\end{acks}



\bibliographystyle{ACM-Reference-Format}
\bibliography{bibtex}

\appendix

\input{appendix}

\end{document}

%% file: macros.tex
\usepackage{bm} 
\usepackage{url}
\usepackage{amsmath}
\usepackage{amsfonts}
\usepackage{amsthm}
\usepackage{algorithmic}
\usepackage{graphicx}
\usepackage{textcomp}
\usepackage{xcolor}
\usepackage{enumerate}
\def\BibTeX{{\rm B\kern-.05em{\sc i\kern-.025em b}\kern-.08em
    T\kern-.1667em\lower.7ex\hbox{E}\kern-.125emX}}

\usepackage{hyperref}
\usepackage[colorinlistoftodos]{todonotes}

\usepackage[ruled, noend]{algorithm2e}

\usepackage{tikz}

\usepackage{microtype}

\newcommand{\set}[1]{\{#1\}}                    
\newcommand{\setof}[2]{\{{#1}\mid{#2}\}}        
\newcommand{\E}{\mathop{\mathbb E}}    

\newcommand{\degree}{\texttt{deg}}

\newcommand{\polylog}{\text{\sf polylog}}

\newtheorem{thm}{Theorem}[section]
\newtheorem{lmm}[thm]{Lemma}
\newtheorem{prop}[thm]{Proposition}
\newtheorem{cor}[thm]{Corollary}

\newtheorem{pbm}{Problem}

\newtheorem{ex}[thm]{Example}

\newtheorem{defn}[thm]{Definition}

\newtheorem{claim}{Claim}

\newcommand{\defeq}{\stackrel{\text{def}}{=}}

\newcommand{\N}{\mathbb N} 
\newcommand{\R}{\mathbb R} 
\newcommand{\Rp}{{\mathbb R}_{\tiny +}} 

\newcommand{\floor}[1]{\lfloor{#1}\rfloor}
\newcommand{\ceil}[1]{\lceil{#1}\rceil}

\newcommand{\lp}[1]{||#1||}
\newcommand{\nop}[1]{}

\newcommand{\openclosed}[1]{(#1]}

%% file: abstract.tex
\begin{abstract}
  Estimating the output size of a query is a fundamental yet longstanding problem in database query processing. Traditional cardinality estimators used by database systems can routinely underestimate the true output size by orders of magnitude, which leads to significant system performance penalty. Recently, upper bounds have been proposed that are based on information inequalities and incorporate sizes and max-degrees from input relations, yet they their main benefit is limited to cyclic queries, because they degenerate to rather trivial formulas on acyclic queries.

  We introduce a significant extension of the upper bounds, by incorporating $\ell_p$-norms of the degree sequences of join attributes.  Our bounds are significantly lower than previously known bounds, even when applied to acyclic queries.  These bounds are also based on information theory, they come with a matching query evaluation algorithm, are computable in exponential time in the query size, and are provably tight when all degrees are ``simple''.
%
\end{abstract}

%% file: intro.tex
\section{Introduction}
\label{sec:intro}

Cardinality estimation is a central yet longstanding open problem in database systems. It allows query optimizers to select a query plan that minimizes the size of the intermediate results and therefore the necessary time and memory to compute the query.
Yet traditional estimators present in virtually all database management systems routinely underestimate the true cardinality by orders of magnitude, which can lead to inefficient query
plans~\cite{DBLP:journals/vldb/LeisRGMBKN18,DBLP:journals/pvldb/HanWWZYTZCQPQZL21,DBLP:conf/sigmod/KimJSHCC22}.

The past two decades introduced {\em worst-case upper bounds} on the output size of a join query. The first such bound is the {\em AGM bound}, which is a function of the sizes of the input tables~\cite{DBLP:journals/siamcomp/AtseriasGM13}. It was further refined in the presence of functional
dependencies~\cite{DBLP:journals/jacm/GottlobLVV12,DBLP:conf/pods/KhamisNS16}. A more general bound is the {\em PANDA bound}, which is a function of both the sizes of the input tables and the max degrees of attributes in these tables~\cite{DBLP:conf/pods/Khamis0S17}. 
These are powerful methods as they can be applied to arbitrary joins and compute {\em provable} upper bounds on the query output size, unlike traditional cardinality estimators 
which often severely  underestimate the query output size~\cite{DBLP:journals/pvldb/LeisGMBK015}. 

However, these theoretical bounds have not had practical impact.  One
reason is that most queries in practice are acyclic queries, where
upper bounds become trivial: they simply multiply the size of one
relation with the \emph{maximum} degrees of the joining relations.
This is not new for a practitioner: standard estimators do the same, but use the \emph{average} degrees instead of the max degrees.
A second, related reason, is that they use essentially the same
statistics as existing cardinality estimators: cardinalities and max
or average degrees.  There have been a few implementations under the
name {\em pessimistic cardinality
  estimators}~\cite{DBLP:conf/sigmod/CaiBS19,DBLP:conf/cidr/HertzschuchHHL21},
but their  empirical evaluation showed that they remain less accurate than
other estimators~\cite{DBLP:journals/pvldb/ChenHWSS22,DBLP:journals/pvldb/HanWWZYTZCQPQZL21}.


In this paper we introduce new upper bounds on the query output size
that use $\ell_p$-norms of degree sequences. The {\em degree sequence}
of a graph is the sorted list of the degrees of the nodes,
$d_1 \geq d_2 \geq \cdots$, where $d_1$ the largest degree, $d_2$ the
next largest, etc. The $\ell_p$-norm of a degree sequence is defined
as $(d_1^p + d_2^p + \cdots)^{1/p}$.  Our method computes an upper
bound in terms of $\ell_p$-norms of the degree
sequences of the join columns; to the best of our knowledge, these are
the first upper bounds that use arbitrary $\ell_p$-norms on the
relations. They strictly generalize previous bounds based on
cardinalities and max-degrees~\cite{DBLP:conf/pods/Khamis0S17},
because the $\ell_1$-norm of an attribute $R.A$ is the size
$\sum_i d_i$ of $R$, and the $\ell_\infty$-norm is the max degree
$d_1$ of $A$.  However, our method can use any other norm,
 which leads to a much tighter upper bound.
We follow the standard assumption in cardinality estimation, and
assume that several $\ell_p$-norms are pre-computed, and available during cardinality estimation.

Like the AGM~\cite{DBLP:journals/siamcomp/AtseriasGM13} and the
PANDA~\cite{DBLP:conf/pods/Khamis0S17} bounds,
our method relies on information inequalities.  The computed bound is
the optimal solution of a linear program, and can be computed in time
exponential in the size of the query.  Our method applies to arbitrary
join queries (cyclic or not), but, unlike AGM and PANDA, it leads to
completely new bounds even for acyclic queries, and uses new kinds of
statistics, which makes it more likely for these theoretical bounds to
have impact in practical scenarios.


\subsection{A Motivating Example}

\label{sec:intro:example}

The standard illustration for size upper bounds is the triangle query:
\begin{align}
Q(X,Y,Z) = & R(X,Y) \wedge S(Y,Z) \wedge T(Z,X), \label{eq:triangle:query:intro}
\end{align}
 for  which the AGM bound~\cite{DBLP:journals/siamcomp/AtseriasGM13}
 (based on the $\ell_1$-norm) is:%
\begin{align}
  |Q| \leq & \left(|R|\cdot |S|\cdot|T|\right)^{1/2} \label{eq:intro:ex:agm}
\end{align}
and the PANDA bound~\cite{DBLP:conf/pods/Khamis0S17} (based on the $\ell_1$ and $\ell_\infty$ norms) is:
\begin{align}
  |Q| \leq |R|\cdot\lp{\degree_S(Z|Y)}_\infty\label{eq:intro:ex:panda}
\end{align}
where $\degree_S(Z|Y) = (d_1, d_2, \ldots, d_m)$ is the degree
sequence of $Y$ in $S$, more precisely $d_i$ is the frequency of the
$i$'th most frequent value $Y=y$.
If the $\ell_2$- and $\ell_3$-norms of the degree sequences are also
available, then we can derive new upper bounds, for example:
\begin{align}
  |Q| \leq & \left(\lp{\degree_R(Y|X)}_2^2\cdot\lp{\degree_S(Z|Y)}_2^2\cdot\lp{\degree_T(X|Z)}_2^2\right)^{1/3}\label{eq:intro:ex:lp:1}\\
  |Q| \leq & \left(\lp{\degree_R(Y|X)}_3^3\cdot\lp{\degree_S(Y|Z)}_3^3\cdot|T|^5\right)^{1/6}\label{eq:intro:ex:lp:2}
\end{align}
Assuming the $\ell_1, \ell_2, \ell_3, \ell_\infty$ norms are
precomputed, then all formulas above give us upper bounds on the query
output size, and we can take the minimal one; which one is the
smallest depends on the actual data.

\nop{As explained in Sec.~\ref{sec:comparison}, our new $\ell_p$ bounds can be asymptotically smaller than the $\ell_1$ and $\ell_\infty$ bounds.}



\subsection{Problem Definition}

Before we define the problem investigated in this paper, we introduce the class of queries and the statistics under consideration.

For a number $n$, let $[n] \defeq \set{1,2,\ldots,n}$.  We
use upper case $X$ for variable names, and lower case
$x$ for values of these variables.  We use boldface for sets of
variables, e.g., $\bm X$, and of constants, e.g.,
$\bm x$.

A full conjunctive (or join) query is defined by:
\begin{align}
  Q(\bm X) = \bigwedge_{j\in[m]} R_j(\bm Y_j) \label{eq:full:cq}
\end{align}
where $\bm Y_j$ is the tuple of variables in $R_j$ and $\bm X=\bigcup_{j\in[m]} \bm Y_j$ is the set of $n \defeq |\bm X|$ variables in the query $Q$.

For a relation $S$ and subsets $\bm U, \bm V$ of its attributes, let
$\degree_S(\bm V|\bm U)$ be the degree sequence of $\bm U$ in the
projection $\Pi_{\bm U\bm V}S$.  Formally, let
$G \defeq (\Pi_{\bm U}(S),$ $\Pi_{\bm V}(S),$ $E)$ be the bipartite
graph whose edges $E$ are all pairs
$(\bm u,\bm v) \in \Pi_{\bm U \bm V}(S)$.  Then
$\degree_S(\bm V|\bm U) \defeq (d_1,d_2,\ldots,d_m)$ is the degree
sequence of the $\bm U$-nodes of the graph. \nop{wlog,
  $d_1 \geq d_2 \geq \ldots \geq d_m$. We write
  $\degree_S(\bm V|\bm U)$ instead of $\degree_S(\bm U)$, because we
  want to allow $\bm V$ to be any subset of attributes of $S$.}

\nop{
For a relation instance $R$ with attributes $\bm X$, and
$\bm U, \bm V \subseteq \bm X$, we let
$\degree_R(\bm V|\bm U=u) = |\Pi_{\bm V}(\sigma_{\bm U=\bm u}(R))|$
and denote by $\degree_R(\bm V|\bm U)$ the decreasingly ordered
sequence
\begin{align*}
  \degree_R(\bm V|\bm U) \defeq & \left(\degree_R(\bm V|\bm U=\bm u_1)\geq \degree_R(\bm V|\bm U=\bm u_2) \geq \cdots\right)
\end{align*}
where $\set{\bm u_1, \bm u_2,\ldots} = \Pi_{\bm U}(R)$.
}

Fix $\bm X$ a set of variables.  An {\em abstract conditional}, or
simply {\em conditional}, is an expression of the form
$\sigma = (\bm V|\bm U)$.  We say that $\sigma$ is {\em guarded} by a
relation $R(\bm Y)$ if $\bm U, \bm V \subseteq \bm Y$; then we write
$\degree_R(\sigma) \defeq \degree_R(\bm V|\bm U)$.  An {\em abstract
  statistics} is a pair $\tau = (\sigma, p)$, where
$p \in \openclosed{0,\infty}$.  If $B \geq 1$ is a real number, then
we call the pair $(\tau, B)$ a {\em concrete statistics}, and call
$(\tau, b)$, where $b \defeq \log B$, a {\em concrete log-statistics}.
If $R$ is a relation guarding $\sigma$, then we say that $R$ {\em
  satisfies} $(\tau, B)$ if $\lp{\degree_R(\sigma)}_p \leq B$.  When
$p=1$ then the statistics is a cardinality assertion on
$|\Pi_{\bm U\bm V}(R)|$, and when $p=\infty$ then it is an assertion
on the maximum degree.  We write $\Sigma = \set{\tau_1,\ldots,\tau_s}$
for a set of abstract statistics, and $\bm B = \set{B_1, \ldots, B_s}$
for an associated set of real numbers; thus, every pair
$(\tau_i, B_i)$ is a concrete statistics.  We will call the pair
$(\Sigma, \bm B)$ a {\em set of (concrete) statistics}, and call
$(\Sigma, \bm b)$, where $b_i\defeq \log B_i$, a set of concrete {\em
  log-statistics}.  We say that $\Sigma$ is guarded by a relational
schema $\bm R = (R_1,\dots, R_m)$ if every $\tau_i \in \Sigma$ has a
guard $R_{j_i}$,
and we say that a database instance $\bm D = (R_1^D, \ldots, R_m^D)$
{\em satisfies} the statistics $(\Sigma, \bm B)$, denoted by $\bm D \models (\Sigma, \bm B)$, if $||\degree_{R^D_{j_i}}(\sigma)||_{p_i}$ $\leq B_{\tau_i}$ for all $\tau_i = (\sigma_i,p_i)\in \Sigma$, where $R_{j_i}$ is the guard of $\sigma_i$. 
We can now state the problem investigated in this paper:

\begin{pbm} \label{problem:upper:bound}
  Given a join query $Q$ and a set of statistics $(\Sigma, \bm B)$ guarded by the (schema
  of the) query $Q$, find a bound $U \in \R$ such that for all
  database instances $\bm D$, if $\bm D \models (\Sigma, \bm B)$, then
  $|Q(\bm D)|\leq U$.
\end{pbm}

The bound $U$ is {\em tight}, if there exists a database instance $\bm D$ such that $\bm D \models (\Sigma, \bm B)$ and $U = O(|Q(\bm D)|)$.

\subsection{Main Results}

\label{sec:main:results}

We solve Problem~\ref{problem:upper:bound} for arbitrary join queries
$Q$, databases $\bm D$ with relations of arbitrary arities, and
statistics $(\Sigma, \bm B)$ consisting of arbitrary $\ell_p$-norms of
degree sequences.
We make the following contributions.

\paragraph{Contribution 1: $\ell_p$ Bounds on Query Output Size.}
Our key observation is that the concrete statistics $|| \degree (\bm V | \bm U) ||_p \leq B$ implies the following inequality in information theory:
\begin{align}
    \frac{1}{p}h(\bm U) + h(\bm V | \bm U) \leq \log B \label{eq:main-inequality}
\end{align}
where $h$ is the entropy of some probability distribution on $R$
(reviewed in Sec.~\ref{sec:background}).
Using~\eqref{eq:main-inequality} we prove the following general upper
bound on the size of the query's output:

\begin{thm} \label{th:main:bound} Let $Q$ be a full conjunctive
  query~\eqref{eq:full:cq}, $\bm U_i, \bm V_i\subseteq \bm X$ be sets
  of variables, for $i\in[s]$, and suppose that the following
  information inequality is valid for all entropic vectors $\bm h$
  with variables $\bm X$:
  \begin{align}
    \sum_{i\in[s]} w_i \left(\frac{1}{p_i}h(\bm U_i) + h(\bm V_i|\bm U_i)\right) \geq h(\bm X)\label{eq:ii:lp}
  \end{align}
  where $w_i \geq 0$, and $p_i \in \openclosed{0,\infty}$, for all
  $i\in[s]$.
  Assume that each conditional $(\bm V_i|\bm U_i)$ in~\eqref{eq:ii:lp}
  is guarded by some relation $R_{j_i}$ in $Q$.  Then, for any
  database instance $\bm D = (R_1^D, R_2^D, \ldots) $, the following
  upper bound holds on the query output size:
  \begin{align}
    |Q(\bm D)| \leq \prod_{i\in[s]} \lp{\degree_{R_{j_i}^D}(\bm V_i|\bm U_i)}_{p_i}^{w_i}\label{eq:bound:lp}
  \end{align}
\end{thm}

We prove the theorem in Sec.~\ref{sec:upper:bound}.
%
%
%
Thus, one approach to find an upper bound on the query output is to
find an inequality of the form~\eqref{eq:ii:lp}, prove it using
Shannon inequalities, then conclude that~\eqref{eq:bound:lp} holds.
For example, the
bounds~\eqref{eq:intro:ex:lp:1}-\eqref{eq:intro:ex:lp:2} stated in our
motivating example follow from the following inequalities:
\begin{align}
  (h(X)+2h(Y|X)) +  (h(Y)+2h(Z|Y))&+ \nonumber\\
  +  (h(Z) + 2h(X|Z)) \geq &  3h(XYZ) \label{eq:intro:ex:ii:1}\\
  (h(X)+3h(Y|X)) +  (h(Z)+3h(Y|Z))&+\nonumber\\
  +  5h(XZ) \geq &6h(XYZ) \label{eq:intro:ex:ii:2}
\end{align}
These can be proven by observing that they are sums of basic Shannon
inequalities (reviewed in Sec.~\ref{sec:background}):
\begin{align*}
    Eq.~\eqref{eq:intro:ex:ii:1} \text{ is sum of } \begin{cases}
      h(X) + h(Y|X) + h(Z|Y) \geq & h(XYZ) \\
      h(Y) + h(Z|Y) + h(X|Z) \geq & h(XYZ) \\
      h(Z) + h(X|Z) + h(Y|X) \geq & h(XYZ)
    \end{cases}\\
    Eq.~\eqref{eq:intro:ex:ii:2} \text{ is sum of } \begin{cases}
       2h(XZ)+2h(Y|X)\geq& 2h(XYZ) \\
       2h(XZ)+2h(Y|Z)\geq& 2h(XYZ) \\
       h(X)+ h(Y|X) + h(Z)\geq&h(XYZ) \\
       h(Y|Z)+h(XZ)\geq&h(XYZ)
    \end{cases}
\end{align*}

\paragraph{Contribution 2: Asymptotically Tighter Cardinality Upper
  Bounds.}  The AGM and PANDA's bounds also rely on an information
inequality, but use only $\ell_1$ and $\ell_\infty$.  Our novelty is
the extension to $\ell_p$ norms. We show in
Sec.~\ref{sec:cardinality:estimation} that this leads to significantly
better bounds.  Quite suprisingly, we are able to improve
significantly the bounds even for acyclic queries, and even for a
single join.

Preliminary experiments (Appendix~\ref{app:comparison}) with cyclic queries on the SNAP graph datasets~\cite{snapnets} and with acyclic queries on the JOB benchmark~\cite{DBLP:journals/pvldb/LeisGMBK015} show that the upper bounds based on $\ell_p$-norms can be orders of magnitude closer to the true cardinalities than the traditional cardinality estimators (e.g., used by DuckDB) and the theoretical upper bounds based on the $\ell_1$ and $\ell_\infty$ norms only. To achieve the best upper bound with our method, a variety of norms are used in the experiments.



\paragraph{Contribution 3: New Algorithm Meeting the New Bounds.} 
The celebrated Worst Case Optimal Join algorithm runs in time bounded
by the AGM
bound~\cite{DBLP:journals/sigmod/NgoRR13,DBLP:journals/jacm/NgoPRR18}.
A more complex algorithm~\cite{DBLP:conf/pods/Khamis0S17} runs in time
bounded by the PANDA bound.  In Sec.~\ref{sec:algorithm} we describe
an algorithm that runs in time bounded by our new $\ell_p$-bounds.
Any such algorithm must include PANDA's as a special case, because our
bounds strictly generalize PANDA's.  Our new algorithm in
Sec.~\ref{sec:algorithm} consists of reducing the general case to
PANDA.  We do this by repeatedly partitioning each relation $R$ such
that a constraint on $\lp{\degree_R(\bm V|\bm U)}_p$ can be replaced
by two constraints, on $|\Pi_{\bm U}(R)|$ and
$\lp{\degree_R(\bm V|\bm U)}_\infty$.
%
%
The original query becomes a union of queries, one per combination of
parts of different relations. The algorithm then evaluates each of
these queries using PANDA's algorithm.

\paragraph{Contribution 4: Computing the bounds.}  One way to describe
the solution to Problem~\ref{problem:upper:bound} is as follows.
Consider a set of statistics $(\Sigma, \bm B)$.  Any valid information
inequality~\eqref{eq:ii:lp} implies some bound on the query output size,
namely $|Q| \leq \prod_{i\in[s]}B_{j_i}^{w_i}$.  The best bound is
their {\em minimum}, over all valid inequalities~\eqref{eq:ii:lp}; we
denote the log of this minimum by $\textit{Log-U-Bound}$.  This
describes the solution to Problem~\ref{problem:upper:bound} as a
minimization problem.  This approach is impractical, because the
number of valid inequalities is infinite.  In Sec.~\ref{sec:bounds} we
describe an alternative, dual characterization of the upper bound, as
a maximization problem, by considering the following quantity:
\begin{align}
  \textit{Log-L-Bound} = \underset{\bm h \models (\Sigma,\bm b)}{\sup} h(\bm X) \label{eq:main-bound}
\end{align}
where $\bm X$ is the set of all variables in the query $Q$, and
$\bm h$ is required to ``satisfy'' the concrete log-statistics
$(\Sigma,\bm b)$, meaning that inequality~\eqref{eq:main-inequality}
is satisfied for every statistics in $\Sigma$.
Equation~\eqref{eq:main-bound} defines a {\em maximization} problem.
Our fourth contribution is:
\begin{thm}[Informal] \label{th:main:bound:dual:informal} If $\bm h$
  ranges over the same closed, convex cone $K$ in
  both~\eqref{eq:ii:lp} and~\eqref{eq:main-bound}, then
  $\textit{Log-U-Bound}=\textit{Log-L-Bound}$.
\end{thm}
We explain the theorem.  $K$ is used implicitly in~\eqref{eq:ii:lp} to
define when the inequality is {\em valid}, namely when it holds
$\forall \bm h \in K$, and also in~\eqref{eq:main-bound}, as the
range of $\bm h$.  The theorem says that, if $K$ is topologically
closed and convex, then the two quantities coincide.  The special case
of the theorem when $K \defeq \Gamma_n$ is the set of {\em
  polymatroids} and~\eqref{eq:ii:lp} are the Shannon inequalities
appeared implicitly in~\cite{DBLP:conf/pods/Khamis0S17}; the general
statement is new, and it includes the non-trivial case when
$K \defeq \bar \Gamma_n^*$ is the closure of entropic vectors
and~\eqref{eq:ii:lp} are all entropic inequalities.  To indicate which
cone was used, we will use the subscript $K$ in~\eqref{eq:main-bound}.
Theorems~\ref{th:main:bound} and~\ref{th:main:bound:dual:informal} and
the fact that $\bar \Gamma_n^* \subseteq \Gamma_n$ imply:
\begin{align}
    \log |Q| \leq \textit{Log-U-Bound}_{\bar{\Gamma}^*_n} \leq \textit{Log-U-Bound}_{\Gamma_n}
    \label{eqn:bounds:ineq}
\end{align}

Theorem~\ref{th:main:bound:dual:informal} has two important
applications. First, it gives us an effective method for solving
Problem~\ref{problem:upper:bound}, when~\eqref{eq:ii:lp} are
restricted to Shannon inequalities, because in that
case~\eqref{eq:main-bound} is the optimal value of a linear program.
Second, it allows us to study the tightness of the bound, by taking a
deeper look at~\eqref{eqn:bounds:ineq}.  We prove
(Appendix~\ref{app:tight:not:tight}) that the {\em entropic bound},
$\textit{Log-U-Bound}_{\bar{\Gamma}^*_n}$, is asymptotically tight
(which is a weaker notion than tightness), while, in general, the {\em
  polymatroid bound}, $\textit{Log-U-Bound}_{\Gamma_n}$, is not even
asymptotically tight.

\paragraph{Contribution 5: Simple degree sequences.}
The tightness analysis leaves us with a dilemma: the entropic bound
is tight but not computable,
while the polymatroid bound is computable but not tight. We reconcile
them in Sec.~\ref{sec:simple:inequalities}: For {\em simple degree
  sequences}, the two bounds coincide, i.e., they become equal.  A
degree sequence \linebreak$\degree_{R}(\bm V|\bm U)$ is {\em simple} if
$|\bm U|\leq 1$.  Moreover, in this case the bound is tight, in our
usual sense: there exists a database $\bm D$ such that the size of the
query output is
$|Q(\bm D)| \geq c\cdot 2^{\textit{Log-U-Bound}_{\Gamma_n}}$, where
$c$ is a constant that depends only on the query $Q$.  The database
$\bm D$ can be restricted to have a special form, called a {\em normal
  database}.

\paragraph{Closely related work.}
Jayaraman et al.~\cite{DBLP:journals/corr/abs-2112-01003} present a
new algorithm for evaluating a query $Q$ and prove a runtime in terms
of $\ell_p$-norms on degree sequences.  Their result is limited to binary relations
(thus all degrees are simple), to a single  value $p$ for a given
query, and to queries with girth $\geq p+1$. (The girth is
the length of the minimal cycle.)  While their work concerns only the
algorithm, not a bound on the output, one can derive a bound from the
runtime of the algorithm, since the output size cannot exceed the
runtime.  In Appendix~\ref{app:rudra} we describe their bound
explicitly, and show that it is a special case of our
inequality~\eqref{eq:ii:lp}.  For example, for the triangle
query~\eqref{eq:triangle:query:intro} their runtime
is~\eqref{eq:intro:ex:lp:1}, but they cannot
derive~\eqref{eq:intro:ex:lp:2}, because the query graph has girth
$3$, hence they cannot use $\ell_3$.  The authors also notice that the
worst-case instance is not always a product database, as in the AGM
bound, but don't characterize it: our paper shows that this is always
a {\em normal database}.
%

The Degree Sequence Bound (DSB)~\cite{DBLP:conf/icdt/DeedsSBC23} is a
tight upper bound of a query $Q$ in terms of the degree sequences of
its join attributes.  The query $Q$ is restricted to be Berge-acyclic,
which also implies that all degree sequences are simple.  There exists
a 1-to-1 mapping between a degree sequence $d_1\geq \cdots \geq d_m$
and its first $m$ norms $\ell_1, \ldots, \ell_m$ (see
Appendix~\ref{app:lp:degrees}), therefore the DSB and our new bound
could have access to the same information.  Somewhat surprisingly, the
DSB bound can be asymptotically better: the reason is that the 1-to-1
mapping is monotone only in one direction.  We describe this analysis
in Appendix~\ref{app:lp:bound:single:join}.  In practice, both methods
have access to fewer statistics than $m$: the DSB bound uses lossy
compression~\cite{DBLP:journals/pacmmod/DeedsSB23}, while our bound
will have access to only a few $\ell_p$-norms, making the two methods
incomparable.


%% file: comparison.tex
\section{Applications}
\label{sec:comparison}

Before we present the technical details of our results, we discuss two
applications: cardinality estimation and query evaluation.

\subsection{Cardinality Estimation}
\label{sec:cardinality:estimation}

Our main intended application of Theorem~\ref{th:main:bound} is for
pessimistic cardinality estimation: given a query and statistics on
the database, compute an upper bound on the query output size.  A
bound is good if it is as small as possible, i.e. as close as possible
to the true output size.  We follow the common assumption in
cardinality estimation that the statistics are precomputed\footnote{It
  takes $O(N \log N)$ time to compute the degree sequence of an
  attribute $X$ of a relation $R$ of size $N$: sort $R$ by $X$,
  group-by $X$, count, then sort again by the count.} and available at
estimation time.  For example the system may have precomputed the
$\ell_2, \ell_5, \ell_\infty$-norms of $\degree_R(Y|X)$ and the
$\ell_1, \ell_{10}$-norms of $\degree_S(Z|Y)$.  We give several
examples of upper bounds of the from~\eqref{eq:bound:lp} that improve
significantly previously known bounds.  For presentation purposes we
describe all bounds in this section using~\eqref{eq:bound:lp}.  A
system would instead rely on~\eqref{eq:main-bound}, i.e. it will
compute the numerical value of the upper bound by optimizing a linear
program, as we explain in Sec.~\ref{sec:bounds}.
To reduce clutter, in this section we abbreviate $|Q(\bm D)|$ with
$|Q|$, and drop the superscript $D$ from an instance $R^D$ when no
confusion arises.


\begin{ex} \label{ex:single:join} As a warmup we start with a single
  join:
  \begin{align}
    Q(X,Y,Z) = & R(X,Y) \wedge S(Y,Z) \label{eq:one:join:query}
  \end{align}
  Traditional cardinality estimators (as found in
  textbooks~\cite{DBLP:books/daglib/0011128}, see also~\cite{DBLP:journals/pvldb/LeisGMBK015}) use the formula
  \begin{align}
  |Q| \approx & \frac{|R|\cdot|S|}{\max(|\Pi_Y(R)|,|\Pi_Y(S)|)}\label{eq:traditional}
  \end{align}
  Since $\frac{|R|}{|\Pi_Y(R)|}$ is the average degree of $R(X|Y)$,
  \eqref{eq:traditional} is equivalent to
  \begin{align}
  |Q| \approx & \min\left(|S|\cdot \text{avg}(\degree_R(X|Y)),|R|\cdot \text{avg}(\degree_S(Z|Y))\right)\label{eq:traditional:2}
  \end{align}
  Turning our attention to upper bounds, we note that the AGM bound is
  $|R|\cdot |S|$.  A better bound is the PANDA bound, which replaces
  $\text{avg}$ with $\max$ in~\eqref{eq:traditional:2}:
  \begin{align}
    |Q| \leq & \min\left(|S|\cdot \lp{\degree_R(X|Y)}_\infty,\ \ |R|\cdot \lp{\degree_S(Z|Y)}_\infty\right) \label{eq:panda:bound:for:join}
  \end{align}

  Our framework derives several new upper bounds, by using
  $\ell_p$-statistics other than $\ell_1$ and $\ell_\infty$.  We start with
  the simplest:
  \begin{align}
    |Q| \leq & \lp{\degree_R(X|Y)}_2 \cdot \lp{\degree_S(Z|Y)}_2  \label{eq:l2:bound:for:join}
  \end{align}
  The reader may notice that this inequality is Cauchy-Schwartz, but,
  in the framework of Th.~\ref{th:main:bound}, it follows from a
  Shannon inequality:
  \begin{align*}
    \frac{1}{2}\left(h(Y) + 2h(X|Y)\right) + \frac{1}{2}\left(h(Y) +
    2h(Z|Y)\right)
    \geq & h(XYZ)
  \end{align*}
  The inequality can be simplified to $h(Y)+h(X|Y)+h(Z|Y)\geq h(XYZ)$,
  which holds because $h(Y) + h(X|Y) = h(XY)$, $h(Z|Y)\geq h(Z|XY)$,
  and $h(XY) + h(Z|XY) = h(XYZ)$; we review Shannon inequalities in
  Sec.~\ref{sec:background}.  Depending on the data,
  ~\eqref{eq:l2:bound:for:join} can be asymptotically better
  than~\eqref{eq:panda:bound:for:join}.  A simple example where this
  happens is when $Q$ is a self-join, i.e.  $R(X,Y) \wedge
  R(Z,Y)$. Then, the two degree sequences are equal,
  $\degree_R(X|Y)=\degree_R(Z|Y)$, and~\eqref{eq:l2:bound:for:join}
  becomes an equality, because $|Q|=\lp{\degree_R(X|Y)}_2^2$.  Thus,
  \eqref{eq:l2:bound:for:join} is exactly $|Q|$,
  while~\eqref{eq:panda:bound:for:join} continues to be an over
  approximation of $|Q|$, and can be asymptotically worse (see
  Appendix~\ref{app:lp:bound:single:join}).

  A more sophisticated inequality for the join query is the following,
  which holds for all $p, q \geq 0$ s.t.
  $\frac{1}{p}+\frac{1}{q} \leq 1$:
  %
  \begin{align}
    |Q| \leq & \lp{\degree_R(X|Y)}_p\cdot\lp{\degree_S(Z|Y)}_q^{\frac{q}{p(q-1)}}|S|^{1-\frac{q}{p(q-1)}} \label{eq:general:bound:for:join:2}
  \end{align}
  Depending on the concrete statistics on the data, this new bound can
  be much better than both~\eqref{eq:panda:bound:for:join}
  and~\eqref{eq:l2:bound:for:join}.  We prove this bound in
  Appendix~\ref{app:lp:bound:single:join}, where we also use this
  bound to study the connection between our $\ell_p$-bounds on the
  Degree Sequence Bound in~\cite{DBLP:conf/icdt/DeedsSBC23}.

  The new
  bounds~\eqref{eq:l2:bound:for:join}-\eqref{eq:general:bound:for:join:2}
  are just two examples, and other inequalities exist. In
  Appendix~\ref{app:simple-path} we provide some empirical evidence
  showing that, even for a single join, these new formulas indeed give
  better bounds on real data.
\end{ex}

\begin{ex} \label{ex:path-query} In real applications most queries are
  acyclic.  In Appendix~\ref{sec:job:benchmark} we conducted some
  preliminary empirical evaluation on the JOB benchmark consisting of
  33 acyclic queries over the IMDB real dataset, and found that the
  new $\ell_p$-bounds are significantly better than both traditional
  estimators (e.g., used by DuckDB) and pessimistic estimators (AGM, PANDA).  We
  give here a taste of how such a bound might look for a path query of length $n\geq 3$:
    \begin{align*}
      Q(X_1,\ldots,X_n) = \bigwedge_{i\in[n-1]} R_i(X_i,X_{i+1})
    \end{align*}
    Traditional cardinality estimators apply~\eqref{eq:traditional}
    repeatedly; similarly PANDA relies on a straightforward extension
    of~\eqref{eq:panda:bound:for:join}.  Our new approach leads, for
    example, to:
    \begin{align*}
      |Q|^p & \leq  |R_1|^{p-2} \cdot \lp{\degree_{R_2}(X_1|X_2)}_{2}^{2}\cdot\\
\cdot  & \prod_{i=2,n-2} \lp{\degree_{R_i}(X_{i+1}|X_i)}_{p-1}^{p-1} \cdot \lp{\degree_{R_{n-1}}(X_n|X_{n-1})}_p^p
    \end{align*}
  This bound holds for any $p \geq 2$, because of the following
  Shannon inequality (proven in Appendix~\ref{app:bound-path}):
    \begin{align}
      & (p-2) h(X_1X_2) + \left(h(X_2) + 2 h(X_1|X_2)\right)+ \nonumber\\
      & \sum_{i=2,n-2} \left(h(X_i)+(p-1)h(X_{i+1}|X_i)\right)\nonumber \\
      & + \left(h(X_{n-1}) + p h(X_n|X_{n-1})\right) \geq ph(X_1\ldots X_n)\label{eq:ex:path-query:shannon:inequality}
    \end{align}
    %
%
  \end{ex}
  We illustrate several other bounds for the path query in
  Appendix~\ref{app:bound-path}.  To our surprise, when we conducted
  our empirical evaluation in Appendix~\ref{sec:job:benchmark}, we
  found that the system used $\ell_p$-norms from a wide range,
  $p \in \set{1,2,\ldots, 29, \infty}$.  This shows the utility of
  having a large variety of $\ell_p$-norm statistics for the purpose
  of cardinality estimation.  It also raises a theoretical question:
  is it the case that, for every $p$, there exists a query/database,
  for which the optimal bound uses the $\ell_p$-norm?  We answer this
  positively next.


\begin{ex} \label{ex:cycle-query} For every $p$,
  there exists a query and a database instance where the $\ell_p$-norm
  on degree sequences leads to the best upper bound.  Consider the cycle query of length $p+1$:
  \begin{align*}
    Q(X_0, \ldots, X_p) = R_0(X_0, X_1)\wedge \ldots \wedge R_{p-1}(X_{p-1}, X_p) \wedge R_p(X_p, X_0) 
  \end{align*}
  For every $q \in [p]$, the following is an upper bound (generalizing~\eqref{eq:intro:ex:lp:1}):
  \begin{align}
    |Q|\leq &\prod_{i=0,p}\lp{\degree_{R_i}(X_{(i+1)\text{mod}(p+1)}|X_i)}_q^{\frac{q}{q+1}} \label{eq:bound:1:bound:1:infty:bound:q}
  \end{align}
  \noindent The bound follows from a Shannon inequality, which we
  defer to Appendix~\ref{app:bound-cycle}.  In the same appendix we
  also prove that, for any $p$, there exists a database instance where
  the bound~\eqref{eq:bound:1:bound:1:infty:bound:q} for $q:=p$ is the
  theoretically optimal bound that one can obtain by using the
  statistics on all
  $\ell_1, \ell_2, \ldots, \ell_p, \ell_\infty$-norms of all degree
  sequences.
\end{ex}

%% file: algorithm.tex
\subsection{Query Evaluation}
\label{sec:algorithm}

The second application is to query evaluation: we show that, if
inequality~\eqref{eq:ii:lp} holds for all polymatroids, then we can
evaluate the query in time bounded by~\eqref{eq:bound:lp} times a
polylogarithmic factor in the data and an exponential factor in the
sum of the $p$ values of the statistics.
Our algorithm generalizes the PANDA's 
algorithm~\cite{DBLP:conf/pods/Khamis0S17} from $\ell_1$ and
$\ell_\infty$ norms to arbitrary norms.  Recall that PANDA starts from
an inequality of the form~\eqref{eq:ii:lp}, where every $p_i$ is
either $1$ or $\infty$, and computes $Q(\bm D)$ in time
$O\left(\prod_i B_i^{w_i}\right)$ if the database satisfies
$|\Pi_{\bm U_i\bm V_i}(R_{j_i})|\leq B_i$ when $p_i=1$ and
$\lp{\degree_{R_{j_i}}(\bm V_i|\bm U_i)}_\infty \leq B_i$ when
$p_i=\infty$.

Our algorithm uses PANDA as a black box, as follows.  It first
partitions the relations on the join columns so that, within each
partition, all degrees are within a factor of two, and each statistics
defined by some $\ell_p$-norm on the degree sequence of the join
column can be expressed alternatively using only $\ell_1$ and
$\ell_\infty$. The original query becomes a union of queries, one per
combination of parts of different relations. The algorithm then
evaluates each of these queries using PANDA.
We describe next the details of data partitioning and the reduction to PANDA.

Consider a relation $R$ with attributes $\bm X$, and consider a
concrete statistics $(\tau, B)$, where $\tau = ((\bm V|\bm U), p)$.
We say that $R$ {\em strongly satisfies $(\tau, B)$}, in notation
$R \models_s (\tau, B)$, if there exists a number $d > 0$ such that
$\lp{\degree_R(\bm V|\bm U)}_\infty \leq d$ and
$|\Pi_{\bm U}(R)| \leq B^p/d^p$.  If $R \models_s (\tau,B)$ then
$R \models (\tau,B)$ because:
\begin{align}
  \lp{\degree_R(\bm V|\bm U)}_p^p \leq & |\Pi_{\bm U}(R)|\cdot \lp{\degree_R(\bm V|\bm U)}_\infty^p \leq \frac{B^p}{d^p} d^p = B^p \label{eq:b1:binfty}
\end{align}
In other words, $R$ strongly satisfies the $\ell_p$ statistics
$(\tau,B)$ if it satisfies an $\ell_1$ and an $\ell_\infty$ statistics
that imply $(\tau, B)$.  We prove:

\begin{lmm} \label{lemma:panda:algorithm} Fix a join query
  $Q$, and suppose that inequality~\eqref{eq:ii:lp} holds for all
  polymatroids. Let $\Sigma = \setof{(\bm V_i|\bm U_i,p_i)}{i\in [s]}$ be
  the abstract statistics and $w_i\geq 0$ be the coefficients in~\eqref{eq:ii:lp}.  
  If a database $\bm D$ strongly satisfies the concrete statistics
  $(\Sigma,\bm B)$, then the query output $Q(\bm D)$ can be computed in time
  $O\left(\prod_{i\in[s]} B_i^{w_i} \polylog\ N\right)$, where $N$ is the size
  of the active domain of $\bm D$.
\end{lmm}
\begin{proof}
  Since $\bm D$ strongly satisfies the concrete statistics
  $(\Sigma,\bm B)$, we can use~\eqref{eq:b1:binfty} and replace each
  $\ell_p$-statistics with an $\ell_1$ and an $\ell_\infty$
  statistics.  We write $B_i$ as
  $B_i = B_{i,1}^{\frac{1}{p}} \cdot B_{i,\infty}$, such that both
  $|\Pi_{U_i}(R_{j_i}^D)|\leq B_{i,1}$ and
  $\lp{\degree_{R_{j_i}}(\bm V|\bm U)}_\infty \leq B_{i,\infty}$ hold.
  Expand the inequality~\eqref{eq:ii:lp} to
  $\sum_i \frac{w_i}{p_i} h(\bm U_i) + \sum_i w_i h(\bm V_i |\bm U_i)
  \leq h(\bm X)$.  This can be viewed as an inequality of the
  form~\eqref{eq:ii:lp} with $2s$ terms, where half of the terms have
  $p_i=1$ and the others have $p_i=\infty$.  Therefore, PANDA's
  algorithm can use this inequality and run in time: and
  \begin{align*}
  O\left(\prod_{i \in [s]}B_{i,1}^{\frac{w_{i}}{p_{i}}} \cdot \prod_{i\in[s]}B_{i,\infty}^{w_{i}}\cdot\polylog\ N\right)    =& O\left(\prod_{i\in[s]} B_i^{w_i} \polylog\ N\right)
  \end{align*}
\end{proof}

In order to use the lemma, we prove the following:
\begin{lmm} \label{lemma:partition}
  Let $R$ be a relation that satisfies an $\ell_p$-statistics,
  $R \models (((\bm V|\bm U),p),B)$.  Then we can partition $R$ into
  $\ceil{2^p}\log N$ disjoint relations,
  $R = R_1 \cup R_2 \cup \ldots$, such that each $R_i$ strongly
  satisfies the $\ell_p$-statistics,
  $R_i \models_s (((\bm V|\bm U),p),B)$.
\end{lmm}

\begin{proof}
  By assumption, $\lp{\degree_R{\bm V|\bm U}}_p^p \leq B^p$.  First,
  partition $R$ into $\log N$ buckets $R_i$,
  $i=1,\ldots,\ceil{\log N}$, where $R_i$ contains the tuples $t$
  whose $\bm U$-component $\bm u$ satisfies:
  \begin{align*}
    2^{i-1} \leq & \degree_R(\bm V|\bm U = \bm u) = \degree_{R_i}(\bm V|\bm U = \bm u) \leq 2^i
  \end{align*}
  Then $|\Pi_{\bm U}(R_i)|\leq B^p/2^{p(i-1)}$, because:
  \begin{align*}
    B^p \geq & \lp{\degree_R(\bm V|\bm U)}_p^p \geq\lp{\degree_{R_i}(\bm V|\bm U)}_p^p \geq|\Pi_{\bm U}(R_i)| \cdot 2^{p(i-1)}
  \end{align*}
  Second, partition $R_i$ into $\ceil{2^p}$ sets
  $R_{i,1}, R_{i,2}, \ldots$ such that
  $|\Pi_{\bm U}(R_i)|\leq B^p/2^{pi}$.  Then, each $R_{i,j}$ strongly
  satisfies the concrete statistics $(((\bm V|\bm U),p), B)$, and their union
  is $R$.
\end{proof}

Our discussion implies:

\begin{thm}
  There exists an algorithm that, given a join query $Q$, an
  inequality~\eqref{eq:ii:lp} that holds for all polymatroids, and a
  database $\bm D$ satisfying the concrete statistics
  $(\Sigma, \bm B)$, computes the query output $Q(\bm D)$ in time
  $O\left(c \cdot \prod_{i\in[s]} B_i^{w_i} \polylog N\right)$; here
  $c = \prod_{i\in[s]} \ceil{2^{p_i}}$, where $p_1, \ldots, p_s$ are the norms
  occurring in $\Sigma$.
\end{thm}

\begin{proof}
  Using Lemma~\ref{lemma:partition}, for each $\ell_{p_i}$-norm, 
  we partition $\bm D$ into a union of $2^{p_i}$
  databases $\bm D_1 \cup \bm D_2 \cup \ldots$, where each $\bm D_j$
  strongly satisfies $(\Sigma, \bm B)$.  
  Resolving $s$ such norms like this partitions $\bm D$ into  $c$ parts.
  We then apply Lemma~\ref{lemma:panda:algorithm} to each part.
\end{proof}

%% file: prelim.tex
\section{Background on Information Theory}

\label{sec:background}

Consider a finite probability space $(D,P)$, where
$P : D \rightarrow [0,1]$, $\sum_{x \in D} P(x) = 1$, and denote by
$X$ the random variable with outcomes in $D$.  The {\em entropy} of
$X$ is:
\begin{align}
  H(X) \defeq & - \sum_{x \in D} P(x) \log P(x) \label{eq:h}
\end{align}
If $N \defeq |D|$, then $0 \leq H(X) \leq \log N$, the equality
$H(X)=0$ holds iff $X$ is deterministic, and $H(X) = \log N$ holds iff
$X$ is uniformly distributed.  Given $n$ jointly distributed random
variables $\bm X = \set{X_1,\ldots,X_n}$, we denote by
$\bm h \in \Rp^{2^{[n]}}$ the following vector:
$h_\alpha \defeq H(\bm X_\alpha)$ for $\alpha \subseteq [n]$, where
$\bm X_\alpha$ is the joint random variable $(X_i)_{i \in \alpha}$,
and $H(\bm X_\alpha)$ is its entropy; such a vector
$\bm h \in \Rp^{2^{[n]}}$ is called {\em entropic}.  We will blur the
distinction between a vector in $\R_+^{2^{[n]}}$, a vector in
$\R_+^{2^{\bm X}}$, and a function $2^{\bm X} \rightarrow \R_+$, and
write interchangeably $\bm h_\alpha$, $\bm h_{\bm X_\alpha}$, or
$h(\bm X_\alpha)$.  A {\em polymatroid} is a vector
$\bm h \in \Rp^{2^{[n]}}$ that satisfies the following {\em basic
  Shannon inequalities}:
\begin{align}
  h(\emptyset) = & \ 0 \label{eq:emptyset:zero}\\
  h(\bm U\cup \bm V) \geq & \ h(\bm U)\label{eq:monotonicity}\\
  h(\bm U) + h(\bm V) \geq & \ h(\bm U \cup \bm V) + h(\bm U \cap \bm V)\label{eq:submodularity}
\end{align}
The last two inequalities are called called {\em monotonicity} and
{\em submodularity} respectively.

For any set $\bm V \subseteq \set{X_1,\ldots,X_n}$, the {\em step
  function} $\bm h^{\bm V}$ is:
\begin{align}
  h^{\bm V}(\bm U) \defeq & \begin{cases}
                              1 & \mbox{if $\bm V \cap \bm U\neq \emptyset$}\\
                              0 & \mbox{otherwise}
                            \end{cases}\label{eq:step:function}
\end{align}
There are $2^n-1$ non-zero step functions (since
$\bm h^{\emptyset}\equiv 0$).  A {\em normal polymatroid} is a
positive linear combination of step functions.  When $\bm V$ is a
singleton set, $\bm V=\set{X_i}$ for some $i=1,n$, then we call
$\bm h^{X_i}$ a {\em basic modular function}.  A {\em modular}
function is a positive linear combination of
$\bm h^{X_1}, \ldots, \bm h^{X_n}$.  The following notations are used
in the literature: $M_n$ is the set of modular functions, $N_n$ is the
set of normal polymatroids, $\Gamma_n^*$ is the set of entropic
vectors, $\bar \Gamma_n^*$ is its topological closure, and $\Gamma_n$
is the set of polymatroids.  It is known that
$M_n \subset N_n \subset \Gamma_n^* \subset \bar \Gamma_n^* \subset
\Gamma_n \subset \Rp^{2^{[n]}}$, that $M_n, N_n, \Gamma_n$ are
polyhedral cones, $\bar \Gamma_n^*$ is a closed, convex cone, and $\Gamma_n^*$ is not
a cone.\footnote{We refer to~\cite{boyd_vandenberghe_2004} for the
  definitions.}

The {\em conditional} of a vector $\bm h$ is defined as:
\begin{align*}
  h(\bm V|\bm U) \defeq & h(\bm U\bm V) - h(\bm U)
\end{align*}
where $\bm U, \bm V \subseteq \bm X$. If $\bm h$ is a polymatroid,
then $h(\bm V|\bm U)\geq 0$.  If $\bm h$ is entropic and realized by
some probability distribution, then:
\begin{align}
  h(\bm V|\bm U) = & \E_{\bm u}[h(\bm V|\bm U=\bm u)] \label{eq:def:cond:h}
\end{align}
where $h(\bm V|\bm U=\bm u)$ is the standard entropy of the random
variable $\bm V$ conditioned on $\bm U=\bm u$.

An {\em information inequality} is a linear inequality of the form:
\begin{align}
  \bm c \cdot \bm h \geq & 0 \label{eq:ii}
\end{align}
where $\bm c \in \R^{2^{[n]}}$.  Give a set
$K\subseteq \Rp^{2^{[n]}}$, we say that the inequality is {\em valid
  for $K$} if it holds for all $\bm h \in K$; in that case we write
$K \models \bm c \cdot \bm h \geq 0$.  {\em Entropic inequalities} are
those valid for $\Gamma^*$ or, equivalently, for $\bar \Gamma_n^*$: it
is an open problem whether they are decidable.  {\em Shannon
  inequalities} are those valid for $\Gamma_n$ and are decidable in
exponential time.



%% file: main-proof.tex
\section{Proof of  Theorem~\ref{th:main:bound}}
\label{sec:upper:bound}
\label{sec:proof:th:main:bound}

In this section we prove Theorem~\ref{th:main:bound}, by showing that
the information inequality~\eqref{eq:ii:lp} implies an upper bound on
the query output size.  The crux of the proof is
inequality~\eqref{eq:main-inequality}, which we prove below in
Lemma~\ref{lemma:lp:degree:bound}.  It establishes a new connection
between information measures and the $\ell_p$-norm,
Eq.~\eqref{eq:lp:degree:bound} below.

We briefly review connections that are known between database
statistics and information measures.  Let $R$ be a relation instance
with attributes $\bm X$ and with $N$ tuples.  Let
$P : R \rightarrow [0,1]$ be any probability distribution whose
outcome consists of the tuples in $R$, in particular
$\sum_{t \in R} P(t)=1$, and let $h:2^{\bm X} \rightarrow \Rp$ be its
entropic vector.  The following two inequalities connect $\bm h$ to
statistics on $R$:
\begin{align}
  \forall \bm V \subseteq \bm X:&&  h(\bm V) \leq & \log |\Pi_{\bm V}(R)| \label{eq:lp:degree:bound:1}\\
  \forall \bm U, \bm V \subseteq \bm X:&&  h(\bm V|\bm U) \leq & \log \lp{\degree_R(\bm V|\bm U)}_\infty    \label{eq:lp:degree:bound:infty}
\end{align}
Eq.~\eqref{eq:lp:degree:bound:infty} follows
from~\eqref{eq:lp:degree:bound:1}, from the fact that, for all
$\bm u \in \Pi_{\bm U}(R)$,
\begin{align*}
    h(\bm V|\bm U = \bm u) &\leq \log \degree_R(\bm V|\bm U = \bm u) \leq \log \max_{\bm u'}\degree_R(\bm V|\bm U = \bm u')\\
    &= \log \lp{\degree_R(\bm V|\bm U)}_\infty
\end{align*}
and from~\eqref{eq:def:cond:h}.  In addition to these two connections,
Lee~\cite{DBLP:journals/tse/Lee87} also proved a connection between
conditional mutual information and multivalued dependencies, which is
unrelated to our paper.
We prove here a new connection:

\begin{lmm}
\label{lemma:lp:degree:bound}
With the notations above, the following holds:
\begin{align}
\forall p \in \openclosed{0,\infty}: && \frac{1}{p}h(\bm U) + h(\bm V|\bm U) \leq &\log \lp{\degree_R(\bm V|\bm U)}_p
    \label{eq:lp:degree:bound}
\end{align}
\end{lmm}

\begin{proof}
  When $p=\infty$, then~\eqref{eq:lp:degree:bound}
  becomes~\eqref{eq:lp:degree:bound:infty}, so we can assume
  $p<\infty$ and rewrite~\eqref{eq:lp:degree:bound} to:
  \begin{align*}
    h(\bm U) + p h(\bm V|\bm U) \leq &\log \lp{\degree_R(\bm V|\bm U)}_p^p
  \end{align*}
  Assume that $\Pi_{\bm U}(R)$ has $N$ distinct values
  $\bm u_1, \ldots, \bm u_N$, and that each $\bm u_i$ occurs with
  $d_i$ distinct values $\bm V=\bm v$.  In particular,
  $\degree_R(\bm V|\bm U=\bm u_i)=d_i$.  Let
  $P_i\defeq P(\bm U=\bm u_i)$ be the marginal probability of
  $\bm u_i$.  We use the definition of the entropy~\eqref{eq:h} and
  the formula for the conditional~\eqref{eq:def:cond:h} and derive:
  \begin{align*}
    h(\bm U) & + p h(\bm V|\bm U) = \sum_i P_i \log \frac{1}{P_i} + p\sum_i P_i h(\bm V|\bm U=\bm u_i)\\
    \leq & \sum_i P_i \log \frac{1}{P_i} + p\sum_i P_i \log d_i  =  \sum_i P_i \log \frac{d_i^p}{P_i} \\
    \leq & \log \left(\sum_i P_i \frac{d_i^p}{P_i}\right) = \log \sum_i d_i^p = \log \lp{\degree_R(\bm V|\bm U)}_p^p
  \end{align*}
  where the last inequality is Jensen's inequality.
\end{proof}

\begin{proof}[Proof of Theorem~\ref{th:main:bound}]
  Assume that inequality~\eqref{eq:ii:lp} holds for all entropic
  vectors $\bm h$.  Fix a database instance
  $\bm D=(R_1, \ldots, R_m)$.

  Consider the uniform probability distribution over the output
  $Q(\bm D)$, and let $\bm h$ be its entropic vector. By uniformity,
  $h(\bm X) = $ \linebreak$\log |Q(\bm D)|$.  By assumption, every conditional term
  $h(\bm V_i|\bm U_i)$, $i\in [s]$ that occurs in~\eqref{eq:ii:lp} has
  a witness $R_{j_i}$.  From Lemma~\ref{lemma:lp:degree:bound}, we
  have
\begin{align*}
    \frac{1}{p_i}h(\bm U_i) + h(\bm V_i|\bm U_i) \leq  \log \lp{\degree_{R_{j_i}}(\bm V_i|\bm U_i)}_{p_i}
\end{align*}
Using inequality~\eqref{eq:ii:lp} we derive:
\begin{align*}
    \log |Q(\bm D)|=h(\bm X) &\leq
    \sum_{i\in[s]} w_i \left(\frac{1}{p_i}h(\bm U_i) + h(\bm V_i|\bm U_i)\right)\\
    &\leq \sum_{i\in[s]} w_i \log \lp{\degree_{R_{j_i}}(\bm V_i|\bm U_i)}_{p_i}
\end{align*}
This immediately implies the upper bound~\eqref{eq:bound:lp}.
\end{proof}

%% file: lower-bound.tex
\section{Computing the Bound}
\label{sec:bounds}

In this section we prove Theorem~\ref{th:main:bound:dual:informal}.
Recall that the main problem in our paper,
problem~\ref{problem:upper:bound}, asks for an upper bound to the
query's output, given a set of concrete statistics on the database.
So far we have proven Theorem~\ref{th:main:bound}, which says that,
for any valid information inequality of the form~\eqref{eq:ii:lp}, we
can infer {\em some} bound.  The best bound is their minimum, over all
valid inequalities~\eqref{eq:ii:lp}, and depends on the concrete
statistics of the database.  In this section we describe how to
compute the best bound, by using the dual of information inequalities.




Given a vector $\bm h \in \Rp^{2^{[n]}}$ an abstract conditional
$\sigma = (\bm V|\bm U)$, and an abstract statistics
$\tau = (\sigma, p)$, we denote by:
\begin{align}
  h(\sigma) \defeq & h(\bm V|\bm U) & h(\tau) \defeq \frac{1}{p}h(\bm U) + h(\bm V|\bm U)\label{eq:h:sigma:tau}
\end{align}
We say that a vector $\bm h$ {\em satisfies} a concrete log-statistics
$(\tau, b)$ if $h(\tau) \leq b$. Similarly, $\bm h\in \Rp^{2^{[n]}}$
{\em satisfies} a set of concrete log-statistics $(\Sigma, \bm b)$, in
notation $\bm h \models (\Sigma, \bm b)$, if $h(\tau_i)\leq b_i$ for
all $\tau_i \in \Sigma, b_i \in \bm b$.

\begin{defn} \label{def:primal:dual:bounds:k} If $\Sigma=\set{\tau_1,\ldots,\tau_s}$ is a set of
  abstract statistics, then a {\em $\Sigma$-inequality} is an
  information inequality of the form:
  \begin{align}
    \sum_{i \in [s]} w_i h(\tau_i) \geq & h(\bm X) \label{eq:ii:lp:revisited}
  \end{align}
  where $w_i \geq 0$.  Notice that~\eqref{eq:ii:lp} in
  Theorem~\ref{th:main:bound} is a $\Sigma$-inequality.
  
  For $K \subseteq \Rp^{2^{[n]}}$, the log-upper bound and log-lower
  bound of a set of log-statistics $(\Sigma,\bm b)$ are:
  \begin{align}
    \text{Log-U-Bound}_K(\Sigma,\bm b) \defeq & \inf_{\bm w: K\models \text{Eq.~\eqref{eq:ii:lp:revisited}}}\sum_{i\in[s]} w_i b_i\label{eq:u:bound}\\
    \text{Log-L-Bound}_K(\Sigma,\bm b) \defeq & \sup_{\bm h \in K: \bm  h\models (\Sigma,\bm b)}h(\bm X)\label{eq:l:bound}
  \end{align}
\end{defn}

Fix a query $Q(\bm X) = \bigwedge_j R_j(\bm Y_j)$ that guards
$\Sigma$, and assume $K = \bar \Gamma_n^*$: by
Theorem~\ref{th:main:bound}, if a database $\bm D$ satisfies the
statistics $(\Sigma, \bm B)$, then
$\log |Q(\bm D)| \leq \text{Log-U-Bound}_K$, but it is an open problem
whether this bound is computable.  On the other hand,
$\text{Log-L-Bound}_K$ is not a bound, but it has two good properties.
First, when $K=\Gamma_n$, then $\text{Log-L-Bound}_K$ {\em is}
computable, as the optimal value of a linear program: we show this in
Example~\ref{ex:linear:optimization}.  Second, when the optimal vector
$\bm h^*$ of the maximization problem~\eqref{eq:l:bound} is the
entropy of some relation, then we can construct a ``worst-case
database instance'' $\bm D$: we use this in
Sec.~\ref{sec:simple:inequalities}.  We prove that~\eqref{eq:u:bound}
and~\eqref{eq:l:bound} are equal:

\begin{thm} \label{th:u:eq:l} If $K$ is any closed, convex cone, and
  $N_n \subseteq K \subseteq \Gamma_n$ then
  $\text{Log-U-Bound}_K=\text{Log-L-Bound}_K$.
\end{thm}

The special case of this theorem when $K=\Gamma_n$ was already
implicit in~\cite{DBLP:conf/pods/Khamis0S17}.  The proof of the
general case is more difficult, and we defer it to
Appendix~\ref{app:th:u:eq:l}.  Both $\bar \Gamma_n^*$ and $\Gamma_n$
are closed, convex cones, hence the theorem applies to both.  We call
the corresponding bounds the {\em almost-entropic bound} (when
$K = \bar \Gamma_n^*$) and the {\em polymatroid bound} (when
$K = \Gamma_n$) respectively.

There are two important applications of Theorem~\ref{th:u:eq:l}.
First, it gives us an effective algorithm for computing the
polymatroid bound, by computing the optimal value of a linear program:
we used this method in all experiments in
Appendix~\ref{app:comparison}.  We illustrate here with a simple
example.

\begin{ex} \label{ex:linear:optimization}
  Consider the triangle query $Q$ in~\eqref{eq:triangle:query:intro}.
  Assume that we have the following statistics for the relations
  $R, S, T$: (a) their cardinalities, denoted by $B_R, B_S, B_T$,
  whose logarithms are $b_R, b_S, b_T$, (b) the $\ell_2$-norms of all
  degree sequences: (c) the $\ell_3$ norms of all degree sequences.
  Then the polymatroid bound~\eqref{eq:l:bound} can be computed by
  optimizing the following linear program, with 8 variables
  $h(\emptyset),h(X),\ldots,$\linebreak$h(XYZ)$:
  \begin{align*}
    & \text{maximize } h(XYZ), \text{ subject to:}\\
    &\hspace*{1em} h(XY) \leq b_R,\ h(YZ) \leq b_S\ h(XZ)\leq b_T &&\text{// cardinality stats}\\
    &\hspace*{1em} \frac{1}{2}h(X)+h(Y|X)\leq b_{((Y|X),2)}\ \ \ \ \ldots  &&\text{// $\ell_2$-norm stats}\\
    &\hspace*{1em} \frac{1}{3}h(X)+h(Y|X) \leq b_{((Y|X),3)}\ \ \ \ \ldots &&\text{// $\ell_3$-norms stats}\\
    &\hspace*{1em} h(X)+h(XYZ) \leq h(XY)+h(XZ)   &&\text{// Shannon inequalities} \\
    &\hspace*{1em} h(Y)+h(XYZ) \leq h(XY)+h(YZ)\ \ \ \ &&\text{// i.e.~\eqref{eq:emptyset:zero}-\eqref{eq:submodularity}}\\
    &\hspace*{1em} \ldots
  \end{align*}
\end{ex}

The second application of Theorem~\ref{th:u:eq:l} is that it allows us
to reason about the tightness of the bounds.  If we can convert the
optimal $\bm h^*$ in the lower bound~\eqref{eq:l:bound} into a
database, then we have a worst-case instance witnessing the fact that
the bound is tight.  We show in Appendix~\ref{app:tight:not:tight}
that the almost-entropic bound is {\em asymptotically} tight (a weaker
form of tightness), while the polymatroid bound is not tight.
However, we show in the next section that the polymatroid bound is
tight in the special case of simple degrees.

%% file: simple-degrees.tex
\section{Simple Degree Sequences}

\label{sec:simple:inequalities}

Call a conditional $\sigma = (\bm V|\bm U)$ {\em simple} if
$|\bm U|\leq 1$; call a set of abstract statistics $\Sigma$ {\em
  simple} if, for all $(\sigma, p) \in \Sigma$, $\sigma$ is simple.
Simple conditionals were introduced
in~\cite{DBLP:journals/tods/KhamisKNS21} to study query containment
under bag semantics.  We prove here that, when all statistics are
simple, then the polymatroid bound is tight, meaning that there exists
a worst case database $\bm D$ such that the size $|Q(\bm D)|$ of the query output is within a
query-dependent constant of the polymatroid bound.  Recall
(Sec.~\ref{sec:background}) that $N_n$ is the set of normal
polymatroids.  

\begin{thm} \label{th:simple}
  If $\Sigma$ is simple, then
  \begin{align*}
    &\text{Log-U-Bound}_{N_n}(\Sigma,\bm b) = 
    \text{Log-U-Bound}_{\bar \Gamma_n^*}(\Sigma,\bm b) = \text{Log-U-Bound}_{\Gamma_n}(\Sigma,\bm b)
  \end{align*}
\end{thm}

The proof relies on a result in~\cite{DBLP:journals/tods/KhamisKNS21},
see Appendix~\ref{app:proof:th:simple}.  In the rest of the section we
will use the theorem to prove that the polymatroid bound is tight.
For that we prove a lemma.  If $T(\bm X)$ is any relation instance with
attributes $\bm X$, then {\em its entropy}, $\bm h_T$, is the entropic
vector defined by the uniform probability distribution on $T$.
Call the relation $T$ {\em totally uniform} if, for all
$\bm V \subseteq \bm X$, the marginal distribution on $\Pi_{\bm V}(T)$
is also uniform.  Equivalently, it is totally uniform if
$\log|\Pi_{\bm V}(T)| = h_T(\bm V)$ for all $\bm V \subseteq \bm X$.
The lemma below proves that, if $\bm h$ is normal, then it can be
approximated by the entropy of a totally uniform $T$, which we will
call a {\em normal} relation.  Recall from Sec.~\ref{sec:background}
that $\bm h$ is {\em normal} if it is a positive, linear combination
of step functions:
\begin{align}
  \bm h = & \sum_{\bm V \subseteq \bm X} \alpha_{\bm V}\bm h^{\bm V}\label{eq:normal:polymatroid:sum}
\end{align}
where $\bm \alpha_{\bm V}\geq 0$.

\begin{lmm} \label{lemma:normal:h:normal:d} Let $\bm h$ be the normal
  polymatroid in~\eqref{eq:normal:polymatroid:sum}, and let $c$ is the
  number of non-zero coefficients $\alpha_{\bm V}$.  Then there exist
  a totally uniform relation $T(\bm X)$ such that
  $|T| \geq \frac{1}{2^c} 2^{h(\bm X)}$, whose entropy $\bm h_T$
  satisfies $\forall \bm U, \bm V \subseteq \bm X$,
  $h_T(\bm V|\bm U) \leq h(\bm V|\bm U)$.
\end{lmm}

The lemma implies tightness of the polymatroid bound:
\begin{cor}
  If all statistics in $\Sigma$ are simple, the polymatroid bound
  $\text{U-Bound}_{\Gamma_n}$
  ($\defeq 2^{\text{Log-U-Bound}_{\Gamma_n}}$) is tight.
\end{cor}
\begin{proof} Since $N_n$ is polyhedral, we have:
  \begin{align*}
    \text{Log-U-Bound}_{\Gamma_n}=& \text{Log-U-Bound}_{N_n} &&\mbox{by Th.~\ref{th:simple}}\\
    =&\text{Log-L-Bound}_{N_n} &&\mbox{by Th.~\ref{th:u:eq:l}}\\
    =&\max_{\bm h \in N_n: (\Sigma,\bm b) \models \bm h} h(\bm X)&&\mbox{by~\eqref{eq:l:bound}}\\
    =&h^*(\bm X)
  \end{align*}
  where $\bm h^*\in N_n$ is optimal solution to the maximization
  problem.  Let $T(\bm X)$ be the totally uniform relation given by
  Lemma~\ref{lemma:normal:h:normal:d}.  Define the database instance
  $\bm D = (R_1^D, \ldots, R_m^D)$ by setting
  $R_j^D \defeq \Pi_{\bm Y_j}(T)$, for $j=1,m$.  Then $\bm D$
  satisfies the constraints $(\Sigma, \bm B)$, because, by total
  uniformity:
  \begin{align*}
    \log  \lp{\degree_{R_{j_i}}&(\bm V_i|\bm U_i)}_p^p =  \log\left(|\Pi_{\bm U_i}(R_{j_i})| \cdot \left(\text{avg}(\degree_{R_{j_i}}(\bm  V_i|\bm U_i))\right)^p\right)\\
    = & \log\left(|\Pi_{\bm U_i}(T)| \cdot \left(\frac{|\Pi_{\bm U_i\bm V_i}(T)|}{|\Pi_{\bm U_i}(T)|}\right)^p\right)\\
    = & h_T(\bm U_i) + p h_T(\bm V_i|\bm U_i) \leq h^*(\bm U_i) + ph^*(\bm V_i|\bm U_i) \leq b_i
  \end{align*}
  The corollary follows from
  $|Q(\bm D)| = |T| \geq \frac{1}{2^c}2^{h^*(\bm X)}=\frac{1}{2^c}\text{U-Bound}_{\Gamma_n}$.
  proving that the bound is tight.
\end{proof}

In the rest of the section we prove
Lemma~\ref{lemma:normal:h:normal:d}.  Given two $\bm X$-tuples
$\bm x = (x_1, \ldots, x_n)$ and $\bm x' = (x_1',\ldots, x_n')$ their
{\em domain product} is
$\bm x \otimes \bm x' \defeq ((x_1,x_1'),\ldots,(x_n,x_n'))$: it has
the same $n$ attributes, and each attribute value is a pair consisting
of a value from $\bm x$ and a value from $\bm x'$.  Given two
relations $T(\bm X), T'(\bm X)$, with the same attributes, their {\em
  domain product} is
$T \otimes T' \defeq \setof{\bm x \otimes \bm x'}{\bm x \in T, \bm x'
  \in T'}$.  The following hold:
\begin{align}
  |T\otimes T'|= & |T|\cdot |T'| \nonumber\\
  \bm h_{T\otimes T'} = & \bm h_T + \bm h_{T'}\label{eq:domain:product:h}
\end{align}
%

Domain products were first introduced by
Fagin~\cite{DBLP:journals/jacm/Fagin82} (under the name {\em direct
  product}), and appear under various names
in~\cite{GeigerPearl1993,DBLP:journals/ejc/KoppartyR11,DBLP:journals/tods/KhamisKNS21}.

\begin{defn} \label{def:normal:relation} For $\bm V \subseteq \bm X$,
  the {\em basic normal relation $T^{\bm V}_N$} is:
  \begin{align}
    T^{\bm V}_N \defeq & \setof{(\underbrace{k,\cdots,k,}_{\text{attributes in }\bm V}\underbrace{0,\cdots,0}_{\bm X-\bm V})}{k=0,N-1}\label{eq:t:w}
  \end{align}
  A {\em normal relation} is a domain product of basic normal
  relations.
\end{defn}
%
%
\begin{prop}
  (1) $T^{\bm V}_N$ is totally uniform.  (2) Its entropy is
  $\bm h_{T^{\bm V}_N}=(\log N)\cdot \bm h^{\bm V}$, where
  $\bm h^{\bm V}$ is the step function.
\end{prop}
The proof is immediate and omitted.  It follows that every normal
relation is totally uniform, because
$|\Pi_{\bm V}(T \otimes T')| = |\Pi_{\bm V}(T)|\cdot|\Pi_{\bm
  V}(T')|=2^{h_T(\bm V)} \cdot 2^{h_{T'}(\bm V)} = 2^{h_{T \otimes
    T'}(\bm V)}$, and the entropy of a normal relation is a normal
polymatroid, because it is the sum of some step functions.  We
illustrate normal relations with an example.

\begin{ex} The following is a basic normal relation:
\begin{align*}
  T^{X,Z}_N = &
                \begin{array}{|ccc|}\hline
                  X&Y&Z \\ \hline
                  0&0&0 \\
                  1&0&1 \\
                  2&0&2 \\
                  \multicolumn{3}{|c|}{\cdots}\\
                  N-1&0&N-1 \\ \hline
                \end{array}
\end{align*}
Its entropy is $(\log N)\bm h^{X,Z}$.  The following are normal
relations:
  \begin{align*}
    T_1 = & \setof{(i,j,k)}{i,j,k \in [0:N-1]} & = & T^X_N\otimes T^Y_N \otimes T^Z_N \\
    T_2 = & \setof{(i,i,i)}{i \in [0:N-1]} & = & T^{XYZ}_N \\
    T_3 = & \setof{(i,(i,j),j)}{i,j\in [0:N-1]} & = & T^{X,Y}_N \otimes T^{Y,Z}_N
  \end{align*}
  Their cardinalities are $|T_1|=N^3$, $|T_2| = N$, $|T_3|=N^2$.
\end{ex}

\begin{proof} (of Lemma~\ref{lemma:normal:h:normal:d}) Fix a normal
  polymatroid $\bm h$ given by~\eqref{eq:normal:polymatroid:sum}.  For
  each $\bm V \subseteq \bm X$, define
  $\beta_{\bm V} \defeq \log\floor{2^{\alpha_{\bm V}}}$.  Then
  $2^{\beta_{\bm V}}$ is an integer, and satisfies the following: (a)
  $\beta_{\bm V} \leq \alpha_{\bm V}$, (b)
  $2^{\beta_{\bm V}} \geq \frac{1}{2} 2^{\alpha_{\bm V}}$ when
  $\alpha_{\bm V}\neq 0$ and $\beta_{\bm V}= \alpha_{\bm V}$ when
  $\alpha_{\bm V|}=0$.  Define the normal relation
  $T \defeq \bigotimes_{\bm V \subseteq \bm X} T^{\bm
    V}_{2^{\beta_{\bm V}}}$; thus, $T$ is uniform.  We check that $T$
  satisfies the lemma.  Its entropy is
  \begin{align*}
    \bm h_T = & \sum_{\bm V \subseteq \bm X} \beta_{\bm V} \bm h^{\bm V}
  \end{align*}
  Condition (1) follows form property (a). For all
  $\bm U, \bm W \subseteq \bm X$:
  \begin{align*}
    h_T(\bm W|\bm U) = & \sum_{\bm V \subseteq \bm X} \beta_{\bm V} h^{\bm V}(\bm W|\bm U)
    \leq   \sum_{\bm V \subseteq \bm X} \alpha_{\bm V} h^{\bm V}(\bm  W|\bm U)=h(\bm W|\bm U)
  \end{align*}
  Condition (2) follows from property (b):
  \begin{align*}
    2^{h_T(\bm X)} = |T| = & \prod_{\bm V \subseteq \bm X} |T^{\bm V}_{2^{\beta_{\bm V}}}|
    =  \prod_{\bm V \subseteq \bm X} 2^{\beta_{\bm V}} 
    \geq   \frac{1}{2^c} \prod_{\bm V \subseteq \bm X} 2^{\alpha_{\bm  V}}=\frac{1}{2^c}2^{h(\bm X)}
  \end{align*}
\end{proof}
%
%
%

\begin{ex} \label{ex:normal:database} Recall that tightness of the AGM
  bound ($\ell_1$-bound) is achieved by a {\em product} database,
  where each relation is the cartesian product of its attributes.  We
  show a query where no product database matches the $\ell_p$-upper
  bound, instead a normal database is needed:
  \begin{align*}
    Q(X,Y,Z) = &R_1(X,Y)\wedge R_2(Y,Z)\wedge R_3(Z,X) \wedge S_1(X) \wedge S_2(Y) \wedge S_3(Z)
  \end{align*}
  Assume the statistics assert that each of $\lp{\degree_{R_1}(Y|X)}_4^4$,
  $\lp{\degree_{R_2}(Z|Y)}_4^4$, $\lp{\degree_{R_3}(X|Z)}_C^4$, $|S_1|$, $|S_2|$,
  $|S_3|$ is $\leq B \defeq 2^b$.  The log-statistics are:
  {\small
    \begin{align}
    &h(X) \leq  b \ \ \ \ h(Y) \leq b\ \ \ \ h(Z) \leq b\nonumber\\
    &h(X)+4h(Y|X) \leq  b\ \ h(Y)+4h(Z|Y)\leq b\ \ h(Z)+4h(X|Z)\leq b\label{eq:ex:normal:database}
    \end{align}
  }
  The following Shannon inequality (see Appendix~\ref{app:proof:th:simple}):
  \begin{align}
    &h(X)+h(Y)+h(Z)+(h(X)+4h(Y|X))+\nonumber\\
    &(h(Y)+4h(Z|Y))+(h(Z)+4h(X|Z))\geq 6h(XYZ) \label{eq:ex:normal:database:inequality}
  \end{align}
  implies $|Q(\bm D)| \leq B$.  To compute the worst-case instance
  $\bm D$, observe that $\bm h^* = b\cdot \bm h^{\set{X,Y,Z}}$ is the
  optimal solution to~\eqref{eq:ex:normal:database}, since it
  satisfies~\eqref{eq:ex:normal:database} and $h^*(XYZ)=b$, and
  define:
  \begin{align*}
    T \defeq  &\setof{(k,k,k)}{k=0,\floor{2^b}-1}
  \end{align*}
  Then $\bm D$ consists of projections of $T$, e.g.
  $R_1^D=\Pi_{XY}(T)$, $S_1^D=\Pi_X(T)$, etc, and
  $|Q(\bm D)| = |T| = \floor{2^b}\geq \frac{1}{2}2^b=\frac{1}{2}B$.
  On the other hand, for any product database $\bm D$, the output
  $Q(\bm D)$ is asymptotically smaller than $B$.  Such a database has
  $R_1^D=[N_X] \times [N_Y]$ and
  $\lp{\degree_{R_1}(Y|X)}_4^4 = N_XN_Y^4$.  The concrete
  $\ell_4$-statistics become:
  \begin{align*}
    N_XN_Y^4 \leq & B & N_YN_Z^4 \leq & B & N_ZN_X^4 \leq & B
  \end{align*}
  By multiplying them we derive $N_X N_Y N_Z \leq B^{3/5}$.  Since
  $Q(\bm D)=[N_X]\times [N_Y] \times [N_Z]$ we derive $|Q(\bm D)|\leq
  B^{3/5}$, which is asymptotically smaller than the upper bound $B$.
\end{ex}


%% file: conclusions.tex
\section{Conclusions}

We have described a new upper bound on the size of the output of a
multi-join query, using $\ell_p$-norms of degree sequences.  Our
techniques are based on information inequalities, and extend prior
results
in~\cite{DBLP:journals/siamcomp/AtseriasGM13,DBLP:journals/jacm/GottlobLVV12,DBLP:conf/pods/KhamisNS16,DBLP:conf/pods/Khamis0S17}.
This is complemented by a query evaluation algorithm whose runtime matches the size bound.
The bound can be computed by optimizing a linear program whose
size is exponential in the size of the query. 
This bound is tight in the case when all degree sequences are simple.

Our new bounds significantly extend the previously known upper
bounds, especially for acyclic queries.  We have also conducted some
very preliminary experiments on real datasets in
Appendix~\ref{app:comparison}, which showed significantly better upper
bounds for acyclic queries than the AGM and PANDA bounds from prior work.  

In future work, we will incorporate our $\ell_p$-bounds into a cardinality estimation system.

%% file: appendix.tex
\input{appendix-1}

\input{appendix-3}
\input{appendix-4}

\input{appendix-5}
\input{appendix-6}

%% file: appendix-1.tex
\section{Equivalence of $L_p$-Norms and Degree Sequences}
\label{app:lp:degrees}

The following is a standard result, establishing a 1-to-1
correspondence between a sequence of length $m$ and its first $m$
norms.  We include it here for completeness.

\begin{lmm}
  \label{lmm:lp-norms=degree-seq} Denote by $\bm S\subseteq \Rp^m$ the
  set of sorted sequences $d_1 \geq d_2 \geq \cdots \geq d_m \geq 0$.
  The mapping $\varphi: \bm S \rightarrow \Rp^m$ defined by
  $\varphi(\bm d) \defeq (\lp{\bm d}_1, \lp{\bm d}_2^2, \ldots,
  \lp{\bm d}_m^m)$ is injective.
\end{lmm}

In other words, having the full degree sequence
$d_1 \geq d_2 \geq \cdots \geq d_m$ is equivalent to having the
$\ell_p$-norms for $p=1,2,\ldots,m$.

\begin{proof}
  We make use of the elementary symmetric polynomials 
  \begin{align*}
      e_0(\bm d) = 1\\
      e_1(\bm d) = d_1+d_2+\ldots+d_m\\
      e_2(\bm d) = \sum_{1\leq i < j \leq m} d_i d_j\\
      \ldots \\
      e_m(\bm d) = d_1 \cdot d_2 \cdot \ldots \cdot d_m\\
      e_k(\bm d) = 0 , k > m.
  \end{align*}

Using Newton's identities (see~\cite{zeilberger1984combinatorial} for a simple proof) we can express the elementary symmetric polynomials using the $L_p$-norms as follows 

\[  k\cdot e_k(\bm d) = \sum_{p=1}^k (-1)^{p-1} e_{k-p}(\bm d) \cdot \|\bm d\|_p^p. \]

Thus, given the values of $\|\bm d\|_p$ for $p\in[m]$,  the first $m$
values of the elementary symmetric polynomials inductively, by:
\begin{align*}
  \bm e_0(\bm d) = & 1 \\
 1\cdot e_1(\bm d) = & e_0(\bm d) \lp{\bm d}_1^1 \\
 2\cdot e_2(\bm d) = & e_1(\bm d) \lp{\bm d}_1^1 - e_0(\bm d)\lp{\bm d}_2^2\\
                   & \ldots \\
  m \cdot e_m(\bm d) = & e_{m-1}(\bm d) \lp{d}_1^1 - e_{m-2}(\bm d) \lp{d}_2^2 + \cdots + (-1)^{m-1}e_0(\bm d) \lp{d}_m^m
\end{align*}
This uniquely determines the values:
\[ e_1(\bm d), e_2(\bm d), \ldots, e_m( \bm d)\]

Using Vieta's formulas we have that the polynomial with roots $d_1, d_2, \ldots, d_m$
corresponds to the polynomial
\[ \lambda^m  - e_1(\bm d) \lambda^{m-1} + e_2(\bm d) \lambda^{m-2}  +\ldots + (-1)^m e_m(\bm d).\]

Thus, the first $m$ symmetric polynomials uniquely determine the degree vector $\bm d$.
\end{proof}

%% file: appendix-3.tex
\section{Relationship to~\cite{DBLP:journals/corr/abs-2112-01003}}

\label{app:rudra}

Jayaraman et al.~\cite{DBLP:journals/corr/abs-2112-01003} consider
conjunctive queries where all relations are binary.  Thus, the query
can be described by a graph with nodes $\bm X$ and edges $E$,
$Q(\bm X) = \bigwedge_{(V,U)\in E}R_{V,U}(V,U)$.  They claim the
following result. Fix a number $p > 1$ and denote by
$L_{V,U} \defeq \lp{\degree_{R_{V,U}(U|V)}}_p$. Consider the following
linear optimization problem:
\begin{align}
  &\text{minimize } \sum_{(V,U)\in E} x_{V,U} \log L_{V,U}\nonumber\\
  \forall U \in \bm X:\   &\sum_{(V,U)\in E} x_{V,U} + \sum_{(U,W)\in E} \frac{x_{U,W}}{p}  \geq 1\label{eq:rudra:lp}\\
  \forall (V,U)\in E:\  &x_{V,U}\geq 0 \nonumber
\end{align}
The authors of~\cite{DBLP:journals/corr/abs-2112-01003} describe an
algorithm that computes the query $Q$ in time
$O(\prod_{(V,U)\in E} L_{V,U}^{x^*_{V,U}})$ (we ignore query-dependent
constants), where $\bm x^*$ is the optimal solution of the program
above.  When $p>2$, then they require the girth of the query graph to
be $\geq p+1$.\footnote{Meaning: the graph has no cycles of length
  $\leq p$.}  No additional condition is required
in~\cite{DBLP:journals/corr/abs-2112-01003} when $p \leq 2$.  However,
the exception for $p\leq 2$ appears to be an omission: the next
example shows that, even for $p = 2$, it is necessary for the graph to
have girth $\geq 3$.

\begin{ex} \label{ex:rudra:p:2} Consider the query
  $Q(U,V) = R(U,V) \wedge S(V,U)$, and $p=2$.  Then
  $x_{U,V} = x_{V,U} = \frac{2}{3}$ is a feasible solution of the
  linear program~\eqref{eq:rudra:lp}, because:
  \begin{align*}
    U:\ & x_{V,U} + \frac{1}{2} x_{U,V} = \frac{2}{3}+\frac{1}{3} = 1\\
    V:\ & x_{U,V} + \frac{1}{2} x_{V,U} = \frac{2}{3}+\frac{1}{3} = 1
  \end{align*}
  The algorithm in~\cite{DBLP:journals/corr/abs-2112-01003} claims to
  compute $Q$ in time $O((L_{V,U}L_{U,V})^{2/3})$.  However, when the
  relations are $R=S=\setof{(i,i)}{i=1,N}$, then
  $L_{V,U} = L_{U,V}=\sqrt{N}$, and the runtime of the algorithm is
  $O(N^{2/3})$, yet the query's output has size $N$, meaning that any
  algorithm requires time $\Omega(N)$.  It appears that, for
  correctness, the algorithm
  in~\cite{DBLP:journals/corr/abs-2112-01003} requires the girth to be
  $\geq p+1$ even when $p=2$.
\end{ex}

Implicit in the result of~\cite{DBLP:journals/corr/abs-2112-01003} is
the claim that the query output size is bounded by
$O(\prod_{(V,U)\in E} L_{V,U}^{x^*_{V,U}})$.  We discuss this upper
bound through the lens of our results.  Consider our upper bound on
the same query, given by~\eqref{eq:ii:lp}:
\begin{align}
  \sum_{(V,U)\in E} x_{V,U} \left(\frac{1}{p}h(V) + h(U|V)\right) \geq & h(\bm X)
\label{eq:rudra:revisited}
\end{align}
We have proven in this paper (Th.~\ref{th:main:bound}) that, {\em if
  the inequality above is valid}, then indeed
$|Q(\bm D)|\leq \prod_{(V,U)\in
  E}\lp{\degree_{R_{V,U}}(V|U)}_p^{x_{V,U}}$.  To check validity
of~\eqref{eq:rudra:revisited} it suffices to check the inequality for
all normal polymatroids, because inequality~\eqref{eq:rudra:revisited}
is simple (Sec.~\ref{sec:simple:inequalities}).  However, the linear
constraints in~\eqref{eq:rudra:lp} check validity only for modular
functions; recall that the modular functions are a strict subset of
the normal polymatroids.  To see this, consider a {\em basic modular
  function} $\bm h^{U_0}$ (Sec.~\ref{sec:background}), where
$U_0 \in \bm X$ is a variable.  Let $U, V \in \bm X$ be any variables.
Then $h^{U_0}(V)=1$ iff $V=U_0$, otherwise $h^{U_0}(V)=0$.  Similarly,
$h^{U_0}(U|V)=1$ iff $U=U_0$; otherwise $h^{U_0}(U|V)=0$.  Also,
$h^{U_0}(\bm X)=1$, because $\bm X$ contains all variables, including
$U_0$.  Therefore the inequality~\eqref{eq:rudra:revisited} applied to
$\bm h^{U_0}$ is:
\begin{align*}
  \sum_{(V,U)\in E}& x_{V,U} \left(\frac{1}{p}h^{U_0}(V) + h^{U_0}(U|V)\right)=\\
&\sum_{(U_0,U)\in E} x_{U_0,U} \left(\frac{1}{p} h^{U_0}(U_0) +  \underbrace{h^{U_0}(U|U_0)}_{=0}\right) + \\
&\sum_{(V,U_0)\in E} x_{V,U_0} \left(\underbrace{\frac{1}{p} h^{U_0}(V)}_{=0} +  h^{U_0}(U_0|V)\right) \\
=& \sum_{(U_0,U)\in E} \frac{x_{U_0,U}}{p} + \sum_{(V,U_0)\in E} x_{V,U_0} \geq h^{U_0}(\bm X)=1.
\end{align*}
%


The first equality above is based on our observation above that
$h^{U_0}(VU)$ and $h^{U_0}(U|V)$ are non-zero iff $U_0$ is one of $U$
or $V$.  Thus, inequality~\eqref{eq:rudra:revisited} is precisely the
constraint in~\eqref{eq:rudra:lp} applied to a basic modular function.
In other words, the result in~\cite{DBLP:journals/corr/abs-2112-01003}
is based on checking the inequality {\em only on modular functions}.
We have seen in Example~\ref{ex:normal:database} that this is
insufficient in general.  In fact, Example~\ref{ex:rudra:p:2} can be
derived precisely in this way, by observing that the following
inequality
\begin{align*}
  \frac{2}{3}(\frac{1}{2}h(V)+h(U|V))+\frac{2}{3}(\frac{1}{2}h(U)+h(V|U))\geq & h(UV)
\end{align*}
is valid for both basic modular functions $\bm h^U$ and $\bm h^V$ (for
example, for $\bm h^U$ the inequality becomes
$\frac{2}{3}(0+1) + \frac{2}{3}(\frac{1}{2}+0)
= 1$), however it fails in general, for example it fails for the
step function $\bm h^{U,V}$:
$\frac{2}{3}(\frac{1}{2}+0)+ \frac{2}{3}(\frac{1}{2}+0) \not\geq 1$.

It turns out, however, that by requiring the girth to be $\geq p+1$,
the implicit claim in~\cite{DBLP:journals/corr/abs-2112-01003} on the
query's upper bound indeed holds.  We state this here explicitly, and
prove it:

\begin{thm} \label{th:rudra:generalized}
  Let $\bm X = \set{X_1, \ldots, X_n}$.  Fix a natural number
  $p \geq 1$ and consider the following inequality:
  \begin{align}
    \sum_{\bm Y \subseteq \bm X} q_{\bm Y} h(\bm Y) + \sum_{i\neq j}r_{ij}(h(X_i) + ph(X_j|X_i)) \geq & k h(\bm X) \label{eq:rudra:generalized}
  \end{align}
  where $k, q_{\bm Y}, r_{ij}$ are natural numbers.  Let $G$ be the
  graph defined by the conditional terms $h(X_j|X_i)$, more precisely,
  $G \defeq (\bm X, E)$, where $E = \setof{(X_i,X_j)}{r_{ij} > 0}$.
  If $G$ has no cycles of length $\leq p$, then the following holds:
  inequality~\eqref{eq:rudra:generalized} holds for all modular
  functions, iff it holds for all polymatroids.
\end{thm}

The theorem implies that if the query's girth is $\geq p+1$ and if
inequality~\eqref{eq:rudra:revisited} holds for all modular functions,
or, equivalently, $\bm x$ is a feasible solution to the linear
program~\eqref{eq:rudra:lp}, then
$|Q(\bm D)| \leq \prod_{(V,U)\in E} L_{V,U}^{x^*_{V,U}}$.  In that
case, by our results in Sec.~\ref{sec:simple:inequalities} (see
Lemma~\ref{lemma:normal:h:normal:d} and its proof) there exists a
normal worst-case product database instance for which the bound is
tight.  The authors in~\cite{DBLP:journals/corr/abs-2112-01003}
already remarked that the worst-case instance for their algorithm is
not always a product database, see Section
1.2.2. in~\cite{DBLP:journals/corr/abs-2112-01003}, however, no
general characterization of the worst-case instance is given.  In our
paper, we have characterized these worst-case instances as being the
{\em normal databases}, which are a natural generalization of product
databases, see Sec.~\ref{sec:simple:inequalities},

For the proof of Theorem~\ref{th:rudra:generalized}, we need the
following {\em modularization lemma}:

\begin{lmm}
  Let $\bm h$ be any polymatroid.  Fix an arbitrary order of the
  variables, say $X_1, X_2, \ldots, X_n$, and define the following
  modular function $\bm h'$:
  \begin{align*}
     h'(X_i) \defeq & h(X_i | X_1 X_2\cdots X_{i-1}),
    & \forall \bm Y\subseteq \bm X:\ h'(\bm Y) \defeq & \sum_{X_i \in \bm Y} h'(X_i)
  \end{align*}
  Then the following hold:
  \begin{align*}
h'(\bm X) = & h(\bm X), & \forall \bm Y \subseteq \bm X: h'(\bm Y) \leq & h(\bm Y), & \forall i<j: h'(X_j|X_i)\leq &h(X_j|X_i)
  \end{align*}
\end{lmm}

\begin{proof}
  The first two claims are well known and we omit their proof.  We
  prove the third claim.  Since $\bm h'$ is modular, we have
  $h'(X_j|X_i) = h'(X_j)$ and the claim follows from
  \begin{align*}
    h'(X_j) = & h(X_j|X_1\ldots X_i \ldots X_{j-1}) \leq h(X_j|X_i)
  \end{align*}
\end{proof}

We now prove Theorem~\ref{th:rudra:generalized}.

\begin{proof} (of Theorem~\ref{th:rudra:generalized}) Denote by
  $E[\bm h]$ the expression on the LHS
  of~\eqref{eq:rudra:generalized}:
  \begin{align*}
    E[\bm h] \defeq &  \sum_{\bm Y \subseteq \bm X} q_{\bm Y} h(\bm Y) + \sum_{i\neq j}r_{ij}(h(X_i) + ph(X_j|X_i)) 
  \end{align*}
  We prove by induction on the number of cycles in the graph $G$
  associated to $E$ the following claim: if $E[\bm h] \geq k h(\bm X)$
  holds for all modular functions $\bm h$, then it holds for all
  polymatroids $\bm h$.

  Base case: $G$ is acyclic.  Let $\bm h$ be any polymatroid.
  Consider any topological order induced by the graph $G$, and let
  $\bm h'$ be the modularization of $\bm h$ induced by this order.  By
  the previous lemma, we have $E[\bm h] \geq E[\bm h']$ and
  $h'(\bm X)=h(\bm X)$.  The inequality $E[\bm h'] \geq k h'(\bm X)$
  holds by assumption, because $\bm h'$ is modular, and this implies
  $E[\bm h] \geq k h(\bm X)$.

  Suppose now that $G$ has a cycle: $U_0, U_1, \ldots, U_{m-1}$, where
  $m \geq p+1$.  Let $\bm h$ be any polymatroid, and write
  $E[\bm h] = E_0[\bm h] + E_1[\bm h]$ where $E_1$ ``is the cycle'':
  \begin{align*}
    E_1&[\bm h] = \sum_{i=0,m-1} \left(h(U_i)+ph(U_{i+1\bmod m}|U_i)\right)\\
    = & \sum_{i=0,m-1} \left(h(U_i) + \sum_{j=0,p-1}h(U_{i+j+1\bmod m}|U_{i+j\bmod m})\right)\\
    \geq & \sum_{i=0,m-1} h(U_i U_{i+1\bmod m} \ldots U_{i+p\bmod m}) \defeq E_1'[\bm h]
  \end{align*}
  In the second line we used the fact that $m \geq p+1$, while in the
  third line we applied a simple Shannon inequality.  Thus, we have
  proven:
  \begin{align}
  E[\bm h] = & E_0[\bm h] + E_1[\bm h] \geq E_0[\bm h] + E_1'[\bm h] \defeq E'[\bm h]\label{eq:local:e:eprime}
  \end{align}

  We claim that the following holds: for every modular function
  $\bm h$, $E'[\bm h] \geq kh(\bm X)$.  To prove the claim, we use the
  fact that, by assumption, $E[\bm h] \geq kh(\bm X)$ holds for all
  modular functions.  In addition, if $\bm h$ is modular, then the
  inequality~\eqref{eq:local:e:eprime} becomes an equality, because,
  for any basic modular function $\bm h^{X}$, $X \in \bm X$, the
  equality $E_1[\bm h^X]=E_1'[\bm h^X]$ holds: indeed if $X$ is one of
  the variables on the cycle, i.e.
  $X \in \set{U_0,U_1,\ldots, U_{m-1}}$, then
  $E_1[\bm h^X]=E_1'[\bm h^X]=1$, otherwise, when $X$ is not on the
  cycle, then both expressions are $=0$.  Thus, we have proven that
  $E[\bm h]= E_0[\bm h] + E_1'[\bm h]$ for all modular functions
  $\bm h$.  Since $E[\bm h] \geq kh(\bm X)$, it follows that
  $E'[\bm h] \geq kh(\bm X)$, proving the claim.

  At this point we apply induction on $E'$.  The graph associated to
  $E'$ has one less cycle than $E$, hence by induction hypothesis, we
  have $E'[\bm h] \geq k h(\bm X)$ for all polymatroids $\bm h$.  It
  follows: $E[\bm h] \geq E'[\bm h] \geq k h(\bm X)$, which completes
  the proof of the theorem.
\end{proof}

%% file: appendix-4.tex
\section{Applications to Pessimistic Cardinality Estimation}
\label{app:comparison}

{\em Pessimistic Cardinality Estimation} refers to a system that
replaces the traditional cardinality estimation module of the query
optimizer with an upper
bound~\cite{DBLP:conf/sigmod/CaiBS19,DBLP:conf/cidr/HertzschuchHHL21,DBLP:journals/pacmmod/DeedsSB23}.
Existing implementations are based on one of two techniques: the AGM
and the PANDA bounds, or a different technique called the {\em degree
  sequence bound}, which applies only to Berge-acyclic queries.  In
this section we extend our discussion in Sec.~\ref{sec:comparison} and
provide both empirical and theoretical evidence for the improvements
provided by the $\ell_p$-bounds. 


\subsection{Preliminary Experiments}
\label{app:simple-path}

We conducted a limited exploration of the usefulness of different $\ell_p$-bounds on (1) eight real datasets representing graphs from the SNAP repository~\cite{snapnets} and (2) the 33 acyclic join queries from the JOB benchmark. We removed the duplicates in the twitter SNAP dataset before processing, the other datasets do not have duplicates.

The goal of an upper bound is to be as close as possible to the true output size of a query.  We computed the upper bound to the true cardinality, for three different
choices of the upper bound: the AGM bound~\cite{DBLP:journals/siamcomp/AtseriasGM13}, the polymatroid
bound from PANDA~\cite{DBLP:conf/pods/Khamis0S17}, and our $\ell_p$-norm based 
bound.  We denoted them by $\set{1}$-bound, $\set{1,\infty}$-bound,
and $\set{1,2,\ldots,p,\infty}$-bound, indicating which $\ell_p$-norms
they used.  We also report the cardinality estimates of DuckDB, a modern publicly-available database management system.

In summary, we found that the bound computed using our approach 
can be significantly tighter than the $\{1\}$-bound and
the $\{1,\infty\}$-bound in our experiment. We also found that DuckDB consistently underestimates the join output size in case of acyclic queries and consistently overestimates in case of the triangle cyclic join query.
Apart from a very few exceptions, provides estimates that are farther away from the true cardinality than our bounds.

\paragraph{Triangle query.} We start with the triangle join query
$Q(X,Y,Z) = R(X,Y)\wedge R(Y,Z)\wedge R(X,Z)$, where $R$ is the edge
relation of the input graph.  Our findings are in the table below:

\begin{center}
\begin{tabular}{l|rrrr}
    \hline
   Dataset  & $\{1\}$  & $\{1,\infty\}$  & $\{2\}$ & DuckDB\\\hline
   ca-GrQc & 32.5 & 15.73 & 3.44 & {\bf 2.99}\\
   ca-HepTh & 69.19 & 19.73 & {\bf 3.80} & 5.17\\
   facebook & 16.26 & 13.74 & {\bf 3.34} & 17.41\\
   soc-Epinions & 101.21 & 101.21 & {\bf 15.27} & 56.03\\
   soc-LiveJournal & 605.54 & 605.54 & {\bf 7.71} & 25.91\\
   soc-pokec & 1765.81 & 1765.81 & {\bf 23.6} & 127.05\\
   twitter & 73.07 & 66.22 & {\bf 4.65} & 36.59\\\hline
\end{tabular}
\end{center}

The numbers represent the ratios between the corresponding upper bound
and the true cardinality: a lower value is better, and 1 is perfect.
Even though we provided all $\ell_p$-norms for $p\in[15]$ and $p=\infty$, the smallest bound was obtained by only using the $\ell_2$-norm. If we were to remove the $\ell_2$-norm, then the next best bound would use the $\ell_3$-norm and be from 1.3 to 4.7 worse, thus still much better than the $\{1\}$-bound and the $\{1,\infty\}$-bound. DuckDB always overestimates in this case of a  cyclic join query; it gives the best estimate in 1/7 datasets by 1.15x relative to our bound. Otherwise, our bound is the best in 6/7 datasets and outperforms DuckDB's estimate by a factor from 1.36x to 7.86x.

\paragraph{One-join query.} We next consider a simple $\alpha$-acyclic query, which is a self-join of the edge relation $R$: $Q(X,Y,Z) = R(X,Y) \wedge R(Y,Z)$.  
The $\{1\}$-bound is $|R|^2$, the $\{1,\infty\}$-bound is $|R|\cdot M$, where $M$ is the minimum of the max-degrees in the first and second column of $R$, while the $\{2\}$-bound is $\left(\lp{\degree_R(X|Y)}_2 \cdot  \lp{\degree_R(Z|Y)}_2\right)$. The table below shows the ratio of each of these three upper bounds to the actual join size:

\begin{center}
    \begin{tabular}{c|rrrr}
    \hline
     Dataset & $\{1\}$ & $\{1,\infty\}$ & $\{2\}$ & DuckDB\\\hline
     ca-GrQc & 2,349.44	& 9.80	& 2.15 & {\bf 0.57}\\
     ca-HepTh & 4,145.99	& 5.19	& {\bf 1.00} & 0.33\\
     facebook & 2,894.12	& 8.23	& 2.45 & {\bf 0.58}\\
     soc-Epinions & 6,485.18	& 22.95	& {\bf 1.75} & 0.25\\
     soc-LiveJournal & 804,671.28 & 162.19 & {\bf 1.45} & 0.22\\
     soc-pokec & 526,733.76 & 150.73 & {\bf 1.27} & 0.34\\
     twitter & 23,374.26 & 15.93 & {\bf 1.61} & 0.34\\\hline
    \end{tabular}
\end{center}

The $\{2\}$-bound is very close (1 - 2.5x larger) to the join output size. DuckDB always underestimates the true cardinality. It gives the best estimates for the ca-GrQc  and facebook datasets (1.22x and 1.44x better than our bounds), otherwise it is worse than the $\{2\}$-bound (by 1.8x to 3.13x). The $\{1,\infty\}$-bound is up to two orders of magnitude higher than the join output size. Finally, the $\{1\}$-bound is from three to six orders of magnitude larger than the join output size. The ratio of 1, i.e., the calculated upper bound which is precisely the join size, is obtained for the edge relation that is symmetric and calibrated with respect to the path query $Q$: This means the degree sequence is the same for both first and second column, on which we join, and there are no dangling tuples that contribute to the $\ell_2$-norm  and not to the join output.

\subsection{Acyclic join queries on the JOB benchmark}

\label{sec:job:benchmark}

Figure~\ref{fig:job-benchmark} shows the ratios of various bounds and estimates to the true cardinality of the query output for each of the 33 join queries in the JOB benchmark\footnote{All 113 JOB queries are variations of these 33 join queries with different selection conditions. Supporting selection conditions is beyond the scope of this paper and subject to on-going work.}. These join queries are over four to 14 relations. Two join queries could not be computed by DuckDB so are excluded. For our approach, we consider statistics for the simple degree sequences on the join columns of each relation and $\ell_p$ norms for $p\in[30]\cup\{\infty\}$. 

Our bounds are always better than the AGM bound by 14 to 53 orders of magnitude and than the PANDA bound by up to three orders of magnitude. DuckDB uses a cardinality estimator that underestimates for all queries by up to five orders of magnitude. Our bounds are better for 24 out of 31 queries (77.41\%), while DuckDB's underestimates are better for 7 out of 31 queries. Whenever DuckDB's estimates are better than ours, this is by a single digit factor (four times under 1.5x, one time 2.44x, one time 4.67x, and one time 6.08x). In contrast, our bounds can be better than DuckDB's estimates by up to four orders of magnitude.

Our bounds use a wide variety of norms and never just the $\ell_1$ and $\ell_\infty$ norms. The queries use from two to seven norms. The $\ell_\infty$ norm is used for all queries. The reason is that they all have many key - foreign key joins that do not increase the size of the query output. The optimal solution of our method uses the $\ell_\infty$ norm on the degree sequence of a primary key column, as each key occurs once so the max-degree is one.

\begin{figure*}[t]
    \begin{center}
    \begin{tabular}{|ll|cr|cc|c|}\hline
     Query \# & \# Relations & Ours & Norms & AGM: $\{1\}$ & PANDA: $\{1,\infty\}$ & DuckDB \\\hline
     1 & 5 & {\bf 1.90E+00} & $\{2,\infty\}$  & 1.01E+15	&  2.56E+00	& 3.50E-01 \\
     2 & 5 & {\bf 1.76E+00} & $\{2,\infty\}$  & 2.70E+22	&  1.22E+01	& 1.34E-01 \\
     3 & 4 & {\bf 1.62E+00}  & $\{2,\infty\}$  & 9.67E+16	&  3.40E+01	& 1.07E-01 \\
     4 & 5 & {\bf 1.43E+00} & $\{1,2,23,\infty\}$  & 2.30E+19	&  1.74E+00	& 2.37E-01 \\
     5 & 5 & {\bf 1.32E+00}  & $\{2,\infty\}$  & 6.57E+14	&  2.07E+01	& 2.15E-01 \\
     6 & 5 & {\bf 2.42E+00} & $\{1,2,3,\infty\}$  & 1.07E+24	&  3.65E+01	& 1.84E-01 \\
     7 & 8 & 2.66E+03 & $\{3,4,5,6,\infty\}$  & 7.80E+32	&  7.35E+04	& {\bf 5.07E-04} \\
     8 & 7 & 1.80E+01 & $\{3,\infty\}$  & 2.90E+31	&  2.97E+03	& {\bf 5.68E-02} \\
     9 & 8 & 3.37E+01 & $\{3,\infty\}$  & 8.27E+37	&  2.69E+03	& {\bf 1.09E-01} \\
     10 & 7 & 3.81E+00 & $\{2,\infty\}$  & 1.12E+26	&  4.51E+01	& {\bf 6.43E-01} \\
     11 & 8 & 1.17E+02 & $\{1,2,6,\infty\}$  & 4.91E+29	&  3.68E+02	& {\bf 1.34E-02} \\
     12 & 8 & {\bf 1.67E+00} & $\{2,3,24,\infty\}$  & 9.14E+27	&  3.14E+01	& 4.34E-02 \\
     13 & 9 & {\bf 1.67E+00} & $\{2,3,24,\infty\}$  & 6.40E+28	&  3.14E+01	& 4.34E-02 \\
     14 & 8 & {\bf 1.93E+00} & $\{2,3,27,\infty\}$  & 4.03E+27	&  4.60E+01	& 1.92E-02 \\
     15 & 9 & {\bf 7.63E+00} & $\{3,6,\infty\}$  & 1.34E+35	&  3.17E+03	& 2.07E-04 \\
     16 & 8 & {\bf 8.36E+01} & $\{3,4,5,\infty\}$  & 4.02E+40	&  1.77E+04	& 2.18E-03 \\
     17 & 7 & {\bf 3.23E+00} & $\{3,\infty\}$  & 4.58E+34	&  2.39E+02	& 1.03E-02 \\
     18 & 7 & {\bf 2.45E+00} & $\{2,3,21,\infty\}$  & 8.92E+28	&  9.23E+01	& 4.97E-02 \\
     19 & 10 & 1.49E+02 & $\{3,4,6,\infty\}$  & 2.07E+45	&  5.90E+04	& {\bf 7.47E-03} \\
     20 & 10 & {\bf 1.14E+01} & $\{2,3,6,\infty\}$  & 6.49E+37	&  1.92E+02	& 2.07E-02 \\
     21 & 9 & {\bf 4.07E+02} & $\{2,4,8,\infty\}$  & 3.42E+34	&  5.07E+03	& 3.31E-04 \\
     22 & 11 & {\bf 2.61E+00} & $\{2,3,4,28,\infty\}$  & 4.56E+38	&  1.99E+02	& 8.89E-04 \\
     23 & 11 & {\bf 3.92E+00} & $\{2,3,4,\infty\}$  & 4.32E+36	&  2.72E+02	& 2.38E-04 \\
     24 & 12 & {\bf 4.53E+02} & $\{3,4,5,8,\infty\}$  & 1.17E+55	&  2.96E+05	& 9.72E-05 \\
     25 & 9 & {\bf 4.26E+00} & $\{3,4,29,\infty\}$  & 7.69E+38	&  7.08E+02	& 9.88E-04 \\
     26 & 12 & {\bf 1.28E+01} & $\{2,3,7,8,25,\infty\}$  & 3.35E+45	&  2.54E+02	& 2.93E-03 \\
     27 & 12 & {\bf 4.30E+03} & $\{2,4,8,\infty\}$  & 3.90E+41	&  5.36E+04	& 9.09E-05 \\
     28 & 14 & {\bf 4.48E+00} & $\{2,3,4,28,\infty\}$  & 8.46E+44	&  3.42E+02	& 3.97E-05 \\
     30 & 12 & {\bf 7.82E+00} & $\{3,4,29,\infty\}$ & 1.53E+45 & 1.30E+03 & 4.30E-05\\
     32 & 6 & 3.27E+01 & $\{1,2,6,\infty\}$  & 7.27E+24	&  5.62E+01	& {\bf 1.86E-01} \\
     33 & 14 & {\bf 9.62E+01} & $\{1,2,4,6,29,30,\infty\}$  & 1.96E+53	&  5.37E+02	& 1.20E-03 \\\hline 
\end{tabular}
\end{center}
\caption{Ratios of various bounds and estimates to the true cardinality of the query output for each of the 33 join queries in the JOB benchmark. Queries 29 and 31 were not computable by DuckDB due to their large output size.}
\label{fig:job-benchmark}
\end{figure*}

\subsection{A Single Join (Example~\ref{ex:single:join})}
\label{app:lp:bound:single:join}

We discuss here in depth our new bounds applied to the single join
query in Example~\ref{ex:single:join}.  For convenience, we repeat
here the query~\eqref{eq:one:join:query}:
\begin{align*}
    Q(X,Y,Z) = R(X,Y)\wedge S(Y,Z)
\end{align*}

\paragraph{Inequality~\eqref{eq:l2:bound:for:join}}
We start by describing a simple example where the
bound~\eqref{eq:l2:bound:for:join} is asymptotically better the PANDA
bound~\eqref{eq:panda:bound:for:join}.  For this purpose we define a
type of database instance that we will also use in the rest of the
section.

\begin{defn} \label{def:alpha:beta}
  An $(\alpha,\beta)$-sequence is a degree sequence of the form:
  \begin{align}
    &(\underbrace{M^\beta, \ldots, M^\beta}_{M^\alpha \mbox{ values}}, \underbrace{1, \ldots, 1}_{M-M^\alpha \mbox{ values}}) \label{eq:alpha:beta:database}
  \end{align}
  where $\alpha, \beta > 0$ and $\alpha+\beta \leq 1$.  An
  $(\alpha,\beta)$ relation is a binary relation $R(X,Y)$ where both
  $\degree_R(Y|X)$ and $\degree_R(X|Y)$ are a
  $(\alpha,\beta)$-sequence.~\footnote{Such a relation exists, either
    by Gale–Ryser theorem, or by direct construction: take $R$ the
    disjoint union of
    $\setof{(i,(i,j))}{i\in[M^\alpha],j\in [M^\beta]}$,
    $\setof{((i,j),i)}{i\in[M^\alpha],j\in [M^\beta]}$, and
    $\setof{(i,i)}{i \in [M-2M^{\alpha+\beta}]}$.}
\end{defn}

In other words, there are $M^\alpha$ nodes with degree $M^\beta$, and
$M-M^\alpha$ nodes with degree 1.

Let both $R$ and $S$ be $(\alpha,\beta)$-instances with
$\alpha=\beta=1/3$.  Then the PANDA
bound~\eqref{eq:panda:bound:for:join} is $M^{4/3}$, while our
bound~\eqref{eq:l2:bound:for:join} is $O(M)$, which is asymptotically
better.

The inequality~\eqref{eq:l2:bound:for:join} is a special case of a
more general inequality, which is of independent interest and we show
it here.  This new inequality uses the number of distinct values in
the columns $R.Y$ and $S.Y$.  Such statistics are often available in
database systems, and they are captured by our framework because any
cardinality statistics is a special case of an $\ell_1$-statistics,
e.g.  $|\Pi_Y(R)|$ is the same as $\lp{\degree_R(X|\emptyset)}_1$.
PANDA also uses such cardinalities: for example, denoting
$M \defeq \min(|\Pi_Y(R)|,$ $|\Pi_Y(S)|)$, PANDA also considers the
following inequality:
\begin{align}
  |Q| \leq & \lp{\degree_R(X|Y)}_\infty \cdot \lp{\degree_S(Z|Y)}_\infty \cdot M,  \label{eq:panda:bound:for:join:2}
\end{align}
Yet the best PANDA bound remains~\eqref{eq:panda:bound:for:join},
because it is always better than~\eqref{eq:panda:bound:for:join:2}.

Our new inequality uses $M$ in the following bound, which holds for
all $p, q > 0$ satisfying $\frac{1}{p}+\frac{1}{q} \leq 1$:
\begin{align}
  |Q| \leq & \lp{\degree_R(X|Y)}_p\cdot\lp{\degree_S(Z|Y)}_q\cdot M^{1-\frac{1}{p}-\frac{1}{q}}  \label{eq:general:bound:for:join}
\end{align}
Inequality~\eqref{eq:l2:bound:for:join} is the special case
of~\eqref{eq:general:bound:for:join} for $p=q=2$, while the PANDA
bound~\eqref{eq:panda:bound:for:join} is the special case
$p=1,q=\infty$ and $p=\infty, q=1$.

We prove~\eqref{eq:general:bound:for:join}, by using the following
Shannon inequality (which is of the form~\eqref{eq:ii:lp}):
\begin{align*}
  \left(\frac{1}{p}h(Y)+ h(X|Y)\right) +
  \left(\frac{1}{q}h(Y)+ h(Z|Y)\right) +& \\
  \left(1-\frac{1}{p}-\frac{1}{q}\right)h(Y)&  \geq h(XYZ)
\end{align*}
The inequality simplifies to $h(Y)+ h(X|Y) + h(Z|Y) \geq h(XYZ)$,
which holds because: $h(Z|Y)\geq h(Z|XY)$; $h(Y) + h(X|Y) = h(XY)$;
and $h(XY) + h(Z|XY) = h(XYZ)$.

Examining closer~\eqref{eq:general:bound:for:join}, we also prove that
it can only be optimal when $\frac{1}{p}+\frac{1}{q}=1$, because,
whenever $p \leq p_1$ and $q \leq q_1$, then the
bound~\eqref{eq:general:bound:for:join} using $(p,q)$ is better than
that using $(p_1,q_1)$.  In particular, the only integral values of
$p,q$ for which~\eqref{eq:general:bound:for:join} could be optimal are
$(\infty,1)$, $(1, \infty)$, and $(2,2)$: other potentially optimal
pairs $(p,q)$ exists, e.g. $(6/5, 6)$, but they require fractional $p$
or $q$. To prove this claim, it suffices to prove the following: if
$p \leq p_1$ then
\begin{align*}
  \frac{\lp{\degree_R(X|Y)}_p}{M^{\frac{1}{p}}}\leq&\frac{\lp{\degree_R(X|Y)}_{p_1}}{M^{\frac{1}{p_1}}}
\end{align*}
We rearrange the inequality as:
\begin{align*}
\lp{\degree_R(X|Y)}_p\leq&\lp{\degree_R(X|Y)}_{p_1}\cdot M^{\frac{1}{p}-\frac{1}{p_1}}
\end{align*}
Denoting $\degree_R(X|Y) = (d_1, d_2, \ldots)$ the inequality becomes:
\begin{align*}
  \left(\sum_i d_i^p\right)^{\frac{1}{p}} \leq & \left(\sum_i d_i^{p_1}\right)^{\frac{1}{p_1}} M^{\frac{1}{p}-\frac{1}{p_1}}
\end{align*}
We raise both sides to the power $p$, and denote by $a_i \defeq d_i^p$
and $q \defeq \frac{p_1}{p}$.  Then the inequality becomes:
\begin{align*}
  \sum_i a_i \leq & \left(\sum_i a_i^q\right)^{\frac{1}{q}} M^{1 - \frac{1}{q}}
\end{align*}
which is H\"older's inequality.  This proves the claim.

\paragraph{Inequality~\eqref{eq:general:bound:for:join:2}} 
Next, we provide the proof of~\eqref{eq:general:bound:for:join:2}, by
establishing the following Shannon inequality:
\begin{align*}
  \left(\frac{1}{p}h(Y)+h(X|Y)\right)+&\left(1-\frac{q}{p(q-1)}\right)h(YZ)\\
  +& \frac{q}{p(q-1)}\left(\frac{1}{q}h(Y)+h(Z|Y)\right) \geq h(XYZ)
\end{align*}
The coefficient $1-\frac{q}{p(q-1)}$ is $\geq 0$ because
$\frac{1}{p}+\frac{1}{q}\leq 1$.  We expand the LHS of inequality and
obtain:
{\footnotesize
\begin{align*}
  \frac{1}{p}h(Y)+&h(X|Y)+h(YZ) - \frac{q}{p(q-1)}h(YZ) + \frac{1}{p(q-1)}h(Y)+\frac{q}{p(q-1)}h(Z|Y)\\
= & \frac{q}{p(q-1)}h(Y)+h(X|Y)+h(YZ)-\frac{q}{p(q-1)}h(YZ)+\frac{q}{p(q-1)}h(Z|Y)\\
= & h(X|Y)+h(YZ) \geq h(XYZ)
\end{align*}
}
which proves the claim.  We will show below
that~\eqref{eq:general:bound:for:join:2} can be strictly better
than~\eqref{eq:general:bound:for:join}.


\paragraph{Comparison to the DSB}
A method for computing an upper bound on the query's output using
degree sequences was described in~\cite{DBLP:conf/icdt/DeedsSBC23},
which uses the full degree sequence $d_1 \geq d_2 \geq \cdots$ instead
of its $\ell_1, \ell_2, \ldots$ norms.  We compare it here to our
method, on our single join query.  It turns out
that~\eqref{eq:general:bound:for:join:2} play a key role in this
comparison.

Suppose $R, S$ have the following degree sequences:
\begin{align*}
  \degree_R(X|Y) = & a_1 \geq a_2 \geq \cdots \geq a_M \\
  \degree_S(Z|Y) = & b_1 \geq b_2 \geq \cdots \geq b_M
\end{align*}
If the system has full access to both degree sequences, then the
Degree-Sequence Bound (DSB) defined
in~\cite{DBLP:conf/icdt/DeedsSBC23} is the following quantity:
\begin{align}
  DSB \defeq & \sum_{i=1,M} a_ib_i \label{eq:dsb:simple}
\end{align}
In general the degree sequences are too large to store, and the DSB
bound needs to use compression~\cite{DBLP:journals/pacmmod/DeedsSB23},
but for the purpose of our discussion we will assume that we know both
degree sequences, and $DSB$ is given by the formula above.  It is easy
to check that $|Q|\leq DSB$.  Our bound~\eqref{eq:l2:bound:for:join}
becomes:
\begin{align*}
  |Q| \leq & \lp{\degree_R(X|Y)}_2\cdot \lp{\degree_S(Z|Y)}_2 = \sqrt{(\sum_i a_i^2)(\sum_i b_i^2)}
\end{align*}
Thus, the $DSB$ and the $\ell_2$-bound above are the two sides of the
Cauchy-Schwartz inequality; $DSB$ is obviously the better one.  $DSB$
is also better than the PANDA bound~\eqref{eq:panda:bound:for:join},
which in our notation is $\min(a_1\sum_i b_i, b_1\sum_ia_i)$ (assuming
$a_1$ and $b_1$ are the largest degrees).  Can we compute a better
$\ell_p$-bound?  We will show that~\eqref{eq:general:bound:for:join:2}
can improve over both~\eqref{eq:panda:bound:for:join},
and~\eqref{eq:l2:bound:for:join}, however, it remains strictly weaker
than the $DSB$ bound.  This may be surprising, given the 1-1
correspondence between the statistics and the $\ell_p$-bounds that we
described in Appendix~\ref{app:lp:degrees}.  The mapping between a
degree sequence of length $M$ and its
$\ell_1, \ell_2, \ldots, \ell_M$-norms is 1-to-1, and, moreover, both
bounds are tight: tightness of the DSB bound was proven
in~\cite{DBLP:conf/icdt/DeedsSBC23}, while tightness of the
polymatroid bound holds because both degrees are simple, and it
follows from our discussion in Sec.~\ref{sec:simple:inequalities}.
So, one expects that some $\ell_p$-bounds will match the $DSB$
expression~\eqref{eq:dsb:simple}.  However, this is not the case, for
a rather subtle reason: it is because the set of databases to which
these two bounds apply, differ.  The 1-to-1 mapping from degrees to
$\ell_p$-norms is monotone in one direction, but not in the other.
For example, consider the degree sequence
$\bm d = (d_1,d_2) = (a+\varepsilon, a-\varepsilon)$, where
$\lp{\bm d}_1 = 2a$, $\lp{\bm d}_2^2 = 2a^2+2\varepsilon^2$.  A
database with degree sequence $\bm d' = (d'_1,d'_2) = (a,a)$ satisfies
the $\ell_p$-constraints, because $\lp{\bm d'}_1 = 2a$,
$\lp{\bm d'}_2^2 = 2a^2$, but it does not satisfy the degree sequence,
because $d'_2 > d_2$.  We show next that the polymatroid bound that we
can obtain from the $\ell_p$-norms can be strictly worse than the DSB.
However, we note that, for practical applications, the degree
sequences in the $DSB$ bound need to be compressed, leading to a
different loss of precision, which makes it incomparable to the
$\ell_p$-bound.

We describe now an instance where there exists a gap between the DSP
bound and the polymatroid bound: the relation $R$ is a
$(0,1/3)$-relation, while $S$ is a $(0,2/3)$-relation, see
Def.~\ref{def:alpha:beta}.  More precisely, the two relations
$R(X,Y), S(Y,Z)$ will have the following degree sequences:
\begin{align*}
  \degree_R(X|Y) = & \left(M^{\frac{1}{3}}, 1, 1, \ldots, 1\right) &&\mbox{$M$ values} \\
  \degree_S(Z|Y) = & \left(M^{\frac{2}{3}}, 1, 1, \ldots, 1\right) &&\mbox{$M$ values}
\end{align*}
There are $M$ degrees equal to 1 in both sequences. The value
$DSB= O(M)$ is asymptotically tight, because $|Q| = O(M)$.  Assume
that we access to all statistics $\lp{\degree_R(X|Y)}_p$,
$\lp{\degree_S(Z|Y)}_p$, for $p=1,2,\ldots,M,\infty$.  We prove:

\begin{claim}
  The polymatroid bound is $M^{\frac{10}{9}}$.
\end{claim}

Normally, the polymatroid bound is computed as the optimal solution of
a linear program, as described in Sec.~\ref{sec:bounds}.  However, to
prove the claim, we proceed differently.  First, we describe an
inequality proving that the polymatroid bound is
$O(M^{\frac{10}{9}})$.  Second, we describe a database instance that
satisfies all the given $\ell_p$-statistics, for which the query
output has size $|Q|=\Omega(M^{\frac{10}{9}})$. These two steps prove
the claim.

We start by computing the $\ell_p$-norms for our instance:
\begin{align*}
  \lp{\degree_R(X|Y)}_p^p = &
                              \begin{cases}
                                O(M) & \mbox{when $p\leq 2$} \\
                                O\left(M^{\frac{p}{3}}\right) &
                                \mbox{when $p \geq 3$}
                              \end{cases} \\
  \lp{\degree_S(Z|Y)}_q^q = &
                              \begin{cases}
                                O(M) & \mbox{when $q= 1$} \\
                                O\left(M^{\frac{2q}{3}}\right) &
                                \mbox{when $q \geq 2$}
                              \end{cases} \\
  |Q| = DSB = & M^{\frac{1}{3}}\cdot M^{\frac{2}{3}} + M = O(M)
\end{align*}

For the first step, we use the
inequality~\eqref{eq:general:bound:for:join:2} specialized for $p=3,
q=2$, which we show here for convenience:\footnote{A direct proof
  follows from the  following Shannon inequality:
\begin{align*}
  \frac{1}{3}\left(h(Y)+3h(X|Y)\right) + \frac{1}{3}h(YZ)+ \frac{1}{3}\left(h(Y) + 2h(Z|Y)\right) \geq & h(XYZ)
\end{align*}
}
\begin{align}
  |Q| \leq & \lp{\degree_R(X|Y)}_3 \cdot |S|^{\frac{1}{3}} \cdot\lp{\degree_S(Z|Y)}_2^{\frac{2}{3}} \label{eq:join:upper:bound:via:solver}
\end{align}

Since $|S| = \lp{\degree_S(Z|X)}_1 = O(M)$, we obtain
\begin{align*}
  |Q| \leq & O\left(M^{\frac{1}{3}} \cdot M^{\frac{1}{3}} \cdot M^{\frac{2}{3} \cdot \frac{2}{3}}\right) = O\left(M^{\frac{10}{9}}\right)
\end{align*}
As a side note, we observe that the other upper
bound~\eqref{eq:general:bound:for:join} leads to strictly larger upper
bounds, for any choice of $p,q$.

For the second step we construct a new database instance $R',S'$ that
satisfies all the  $\ell_p$-statistics that we computed for $R, S$.
We describe them using their degrees:
\begin{align*}
  \degree_{R'}(X|Y) = & O\left(M^{\frac{1}{9}},\ldots, M^{\frac{1}{9}}\right) &&\mbox{$M^{\frac{2}{3}}$ values} \\
  \degree_{S'}(Z|Y) = & O\left(M^{\frac{1}{3}},\ldots, M^{\frac{1}{3}}\right) &&\mbox{$M^{\frac{2}{3}}$ values}
\end{align*}
Then the following hold:
\begin{align*}
  \lp{\degree_{R'}(X|Y)}_p^p = & O\left(M^{\frac{p}{9} + \frac{2}{3}}\right)\\
  \lp{\degree_{S'}(Z|Y)}_q^q = & O\left(M^{\frac{q}{3}+\frac{2}{3}}\right)\\
  |Q'| = & M^{\frac{1}{9}}\cdot M^{\frac{1}{3}}\cdot M^{\frac{2}{3}}=M^{\frac{10}{9}}
\end{align*}
We check that the $\ell_p$-norms of the degrees of $R', S'$ are no
larger than those of $R, S$:
\begin{align*}
  p \leq 2: && \lp{\degree_{R'}(X|Y)}_p^p = & O\left(M^{\frac{p}{9} + \frac{2}{3}}\right)\leq O(M) = \lp{\degree_R(X|Y)}_p^p\\
  p \geq 3: && \lp{\degree_{R'}(X|Y)}_p^p = & O\left(M^{\frac{p}{9} + \frac{2}{3}}\right)\leq O\left(M^{\frac{p}{3}}\right) = \lp{\degree_R(X|Y)}_p^p\\
  \\
  q = 1: &&   \lp{\degree_{S'}(Z|Y)}_1 = & O\left(M^{\frac{1}{3}+\frac{2}{3}}\right)\leq O(M) = \lp{\degree_S(Z|Y)}_1\\
  q \geq 2: &&   \lp{\degree_{S'}(Z|Y)}_q^q = & O\left(M^{\frac{q}{3}+\frac{2}{3}}\right)\leq O\left(M^{\frac{2q}{3}}\right) = \lp{\degree_S(Z|Y)}_q^q
\end{align*}
Similarly, $|R'.Y| = |S'.Y| = M^{\frac{2}{3}} \leq M$.  It follows
that the relations $R', S'$ satisfy all constraints on the
$\ell_p$-norms, including those on $|R'.Y|, |S'.Y|$ (assuming the
latter are available).  Yet the size of the output of the query on
$R', S'$ is $M^{\frac{10}{9}}$.

As explained earlier, the issue stems from the fact that the DSB bound
does not permit the instance $R', S'$, since its degree sequences are
not dominated by those of $R, S$.

\subsection{The Chain Query (Example~\ref{ex:path-query})}
\label{app:bound-path}

We prove that inequality~\eqref{eq:ex:path-query:shannon:inequality}
is a Shannon inequality, by writing it as a sum of the following
inequalities, each which can be verified immediately:
\begin{align*}
  (p-2)h(X_1X_2)+\sum_{i=2,n-2}(p-2)h(X_{i+1}|X_i)+(p-2)h(X_n|X_{n-1}) \geq &\\
  \geq (p-2)h(X_1\ldots X_n)& \\
  h(X_2)+h(X_1|X_2)+\sum_{i=2,n-2}h(X_{i+1}|X_i) +h(X_n|X_{n-1}) \geq &\\
  \geq h(X_1\ldots X_n)&\\
  h(X_1|X_2) + \sum_{i=2,n-2} h(X_i) +  h(X_{n-1})+h(X_n|X_{n-1}) \geq\\
  \geq  h(X_1\ldots X_n) &
\end{align*}

\subsection{The Cycle Query (Example~\ref{ex:cycle-query})}
\label{app:bound-cycle}

We prove here the output
bound~\eqref{eq:bound:1:bound:1:infty:bound:q}, then show that, for
every $p \geq 1$ there exists a database instance where this bound for
$q:=p$ is the theoretically optimal bound that can be derived using
all statistics on $\ell_1, \ell_2, \ldots, \ell_p, \ell_\infty$ norms.

To prove~\eqref{eq:bound:1:bound:1:infty:bound:q}, we show the
following Shannon inequality, where the arithmetic on the indices is
taken modulo $p+1$, i.e. $i+1$ means $(i+1) \mod (p+1)$ etc:
\begin{align}
  \sum_{i=0,p}\left(h(X_i) + qh(X_{i+1}|X_i)\right) \geq& (q+1)h(X_0\ldots X_p)\label{eq:cyclic:shannon:inequality}
\end{align}

To prove the inequality, we proceed as follows.  First, we observe
that, for each $i=0,p$, the following is a Shannon inequality:
{\small
  \begin{align*}
    h(X_i) + h(X_{i+1}|X_i) + \ldots + h(X_{i+q}|X_{i+q-1})\geq & h(X_i X_{i+1} \ldots X_{i+q})
  \end{align*}
}
As before, all indices are taken modulo $p+1$, for example $i+q$ means
$(i+q)\mod (p+1)$.  Each inequality above can be easily checked.
Next, we add up these $p+1$ inequalities, and make two observations.
First, the sum of their LHS is precisely the LHS
of~\eqref{eq:cyclic:shannon:inequality}.  Second, after adding up
their RHS, we use the following Shannon inequality
$\sum_{i=0,p} h(X_i \ldots X_{i+q}) \geq (q+1)h(X_0X_1\ldots X_p)$
(which holds because each variable $X_k$ occurs exactly $q+1$ times on
the left, hence this is a Shearer-type inequality).  Together, these
observations prove~\eqref{eq:cyclic:shannon:inequality}.

We compare now the upper
bound~\eqref{eq:bound:1:bound:1:infty:bound:q} to the AGM and PANDA
bounds.  To reduce the clutter we will assume that
$R_0 = R_1 = \cdots = R_p$.  Then the AGM bound and the PANDA bounds
are:
\begin{align}
   |Q| \leq & |R|^{\frac{p+1}{2}} & |Q| \leq & |R|\cdot \lp{\degree_R(Y|X)}_\infty^{p-1}\label{eq:cycle:agm:panda}
\end{align}
They follow from the following straightforward Shannon inequalities:
{\small
  \begin{align*}
    h(X_0X_1)+ h(X_1X_2) + \cdots + h(X_pX_0) \geq & 2h(X_0X_1X_2\ldots X_p)\\
    h(X_0X_1) + h(X_2|X_1) + \cdots + h(X_p|X_{p-1}) \geq & h(X_0X_1X_2\ldots X_p)\\
  \end{align*}
}

Finally, we describe a database instance for which
bound~\eqref{eq:bound:1:bound:1:infty:bound:q} for $q := p$ is the
best.  The instance consists of the $(\alpha,\beta)$-relation $R$ for
$\alpha=\beta = \frac{1}{p+1}$ (see Def.~\ref{def:alpha:beta}); to
simplify the notations here we will rename $M$ to $N$.  Thus, we have
$|R|=N$, $\lp{\degree_R(Y|X)}_q^q=N$ for $q \in [p]$,
$\lp{\degree_R(Y|X)}_\infty = N^{\frac{1}{p+1}}$, and the bounds
in~\eqref{eq:cycle:agm:panda}
and~\eqref{eq:bound:1:bound:1:infty:bound:q} become
$N^{\frac{p+1}{2}}$, $N^{\frac{2p}{p+1}}$, and $N^{\frac{p+1}{q+1}}$
respectively.  The best bound among them is the latter, when $q = p$,
which gives us $|Q| \leq \lp{\degree_R(Y|X)}_p^p = (1+o(N)) N$.  All
other bounds are asymptotically worse.  Thus, among these three
formulas,~\eqref{eq:bound:1:bound:1:infty:bound:q} is the best, namely
for $q := p$.  However, this does not yet prove that these formulas
provide the best bounds if we have access to the given statistics.

We show now that these bounds are tight.  In other words, we show that
there exists relation instances $R$ for which the bounds are tight, up
to constant factors.  We already know this for the $\set{1}$-bound
(the AGM bound), since the AGM bound is $N^{\rho^*}$, where $\rho^*$
is the optimal fractional edge covering number of the $(p+1)$-cycle,
which is $\rho^*=\frac{p+1}{2}$.

Consider now the $\set{1,\infty}$-bound, in other words we have only
the statistics for $\lp{\degree_R(Y|X)}_1$ (which is $|R|$) and
$\lp{\degree_R(Y|X)}_\infty$.  We prove that the PANDA bound
in~\eqref{eq:cycle:agm:panda} is indeed optimal.  In fact we prove a
more general claim: the $\set{1,\infty}$-bound of the cycle query is
$|Q| \leq N\cdot D^{p-1}$, whenever $|R| \leq N$,
$\lp{\degree_R{Y|X}}_\infty\leq D$, and $N, D$ are numbers satisfying
$D^2 \leq N$.  In our case we have $D = N^{\frac{1}{p+1}}$, and the
claim implies that the $\set{1,\infty}$-bound is
$N^{1+\frac{p-1}{p+1}}=N^{\frac{2p}{p+1}}$.  To prove this claim, we
will refer to the polymatroid upper bound, and polymatroid lower bound
in Def.~\ref{def:primal:dual:bounds:k}.  The Shannon inequality that
we proved in Example~\ref{ex:cycle-query} implies
$\text{Log-U-Bound}_{\Gamma_n}(Q) \leq \log N + (p-1) \log D$.  We
also have
$\text{Log-U-Bound}_{\Gamma_n}(Q)=$$\text{Log-U-Bound}_{N_n}(Q)$ (by
Theorem~\ref{th:simple}), where $N_n$ are the normal polymatroids, and\linebreak
$\text{Log-U-Bound}_{N_n}(Q)=$$\text{Log-L-Bound}_{N_n}(Q)$ by
Theorem~\ref{th:u:eq:l}.  We claim that there exists a normal
polymatroid that satisfies the $\set{1,\infty}$-statistics and where
$h(X_0\ldots X_p) = \log N + (p-1) \log D$: the claim implies
$\log N + (p-1) \log D\leq \text{Log-L-Bound}_{\Gamma_n}(Q)$, which
proves that the $\set{1,\infty}$-bound is $N \cdot D^{p-1}$.  To prove
the claim, consider the following polymatroid:
\begin{align*}
h(\emptyset) = & 0, & \forall \bm W\neq \emptyset, \  h(\bm W) \defeq & \log N + (|\bm W|-2) \log D
\end{align*}
Then $h$ satisfies the required statistics:
\begin{align*}
\forall i: \ \   h(X_iX_{i+1}) \leq & \log N & h(X_{i+1}|X_i) \leq & \log D
\end{align*}
and $h(X_0X_1\ldots X_p) = \log N + (p-1)\log D$.  It remains to
observe that $h$ is a normal polymatroid, which follows by writing it
as
$\bm h = (\log N-2\log D)\cdot h^{\bm X} + \log D \cdot \sum_{i=0,p}
h^{X_i}$.

Finally, we prove that, if we have available all statistics\linebreak
$\lp{\degree_R(Y|X)}_q$ for $q =1, 2, \ldots, p, \infty$, then the
best query upper bound is~\eqref{eq:bound:1:bound:1:infty:bound:q}.
Fix a number $q \in [p]$, and let $N, L, D$ be three positive numbers
satisfying $L \leq N$ and $L \leq D^{q+1}$.  Then we claim that the
$\set{1,2,\ldots,q,\infty}$-bound of the cyclic query in
Example~\ref{ex:cycle-query}, when the input relation satisfies the
statistics $|R| \leq N$, $\lp{\degree_R(Y|X)}_r^r \leq L$, for all
$r\leq q$, and $\lp{\degree_R(Y|X)}_\infty \leq D$, is
$|Q| \leq L^{\frac{(p+1)q}{q+1}}$.  The claim applies to our database
instance $(\alpha,\beta)$ for $\alpha=\beta=\frac{1}{p+1}$, because we
have $L = (1+o(N))N$ and $D = N^{\frac{1}{p+1}}$, and implies that the
$\set{1,2,\ldots,q}$-bound is $L^{\frac{(p+1)q}{q+1}}$.  To prove the
claim, we use the same reasoning as above: it suffices to describe a
polymatroid satisfying the statistics
  \begin{align*}
    h(X_iX_{i+1}) \leq & \log N \\
\forall r=2,q:\ \     h(X_iX_{i+1}) + (r-1)h(X_{i+1}|X_i) \leq &  \log L \\
    h(X_{i+1}|X_i) \leq & \log D
  \end{align*}
  The desired polymatroid is the following modular function:\linebreak
  $h(\bm W) \defeq \frac{|\bm W|\cdot \log L}{q+1}$.  In other words,
  $\bm h = \frac{1}{q+1}\sum_{i=0,p}\bm h^{X_i}$.  Then, the first
  inequality above is $\frac{2\log L}{q+1} \leq \log N$, and it holds
  because $L \leq N$.  The second inequality is
  $(r+1) \frac{\log L}{q+1} \leq \log L$, which holds because
  $r \leq q$.  And the third inequality is
  $\frac{\log L}{q+1} \leq \log D$, which holds by the assumption
  $L \leq D^{q+1}$.

\subsection{Loomis-Whitney Query}

All examples so far used only binary relations.  We illustrate here
some examples with relations of higher arity.  More precisely, we
derive general upper bounds for the class of queries, called
Loomis-Whitney, that have relational atoms with more than two join
variables. A Loomis-Whitney query has $n$ variables and $n$ relational
atoms, such that there is one atom for each set of $n-1$
variables. The triangle query is the Loomis-Whitney query with $n=3$.

The Loomis-Whitney query with $n=4$ is:
\begin{align*}
    Q(X,Y,Z,W) = A(X,Y,Z)\wedge B(Y,Z,W)\wedge C(Z,W,X)\wedge D(W,X,Y)
\end{align*}

One bound that can be obtained with our framework is the following:
\begin{align*}
    |Q|^4 \leq \lp{\degree_A(YZ|X)}_2^2 \cdot |B| \cdot \lp{\degree_C(WX|Z)}_2^2 \cdot |D|
\end{align*}
This bound follows from  the following information inequality:
\begin{align*}
  4 h(XYZW) \leq & (h(X)+2h(YZ|X)) + h(YZW)\\
  + & (h(Z) + 2h(WX|Z))+ h(WXY)
\end{align*}
%
%
The inequality holds because it is a sum of 4 Shannon inequalities:
\begin{align*}
  h(XYZW) \leq & h(X) + h(YZ|X) + h(W|YZ) \\
  h(XYZW) \leq & h(Z) + h(WX|Z) + h(Y|WX) \\
  h(XYZW) \leq & h(YZ|X) + h(WX) \\
  h(XYZW) \leq & h(WX|Z) + h(YZ)
\end{align*}
%
%
%


%% file: appendix-5.tex
\section{Proofs from Sec.~\ref{sec:bounds}}

\subsection{Proof of Theorem~\ref{th:u:eq:l}}

\label{app:th:u:eq:l}

We restate Theorem~\ref{th:u:eq:l} here in an extended form

\begin{thm} \label{th:u:eq:l:app} (1) If $K$ is any closed, convex
  cone,\footnote{We refer to~\cite{boyd_vandenberghe_2004} for the
    definitions.} and $N_n \subseteq K \subseteq \Gamma_n$ then
  $\text{Log-U-Bound}_K=\text{Log-L-Bound}_K$.

  (2) If $K$ is polyhedral (meaning: it has the form
  $K = \setof{\bm h}{\bm M \cdot \bm h \geq 0}$ for some matrix
  $\bm M$) then $\text{Log-U-Bound}_K=\text{Log-L-Bound}_K$, and the
  two bounds can be described by a pair of primal/dual linear
  optimization programs.  In particular, $\inf/\sup$
  in~\eqref{eq:u:bound}-\eqref{eq:l:bound} can be replaced by
  $\min/\max$, meaning that there exists optimal solutions $\bm w^*$
  and $\bm h^*$ that achieve the upper and lower bounds, and these can
  be computed in time exponential in the number of variables of $Q$.
\end{thm}

The {\em weak duality property} states the following:
\begin{align}
  \text{Log-L-Bound}_K(\Sigma,\bm b)\leq & \text{Log-U-Bound}_K(\Sigma,\bm b) \label{eq:l:bound:leq:u:bound}
\end{align}
Weak duality is easy to check: If $\bm h \in K$ satisfies
$\bm h \models (\Sigma, \bm b)$, and
$\bm w = (w_\tau)_{\tau \in \Sigma}$ is such that
$K\models \mbox{Eq.~\eqref{eq:ii:lp:revisited}}$, then
$h(\bm X) \leq \sum_{\tau \in \Sigma} w_\tau h(\tau) \leq \sum_{\tau
  \in \Sigma} w_\tau b_\tau$, and weak duality follows from the fact
that \linebreak$\text{Log-L-Bound}_K = \sup_{\bm h}(\cdots)$ and
$\text{Log-U-Bound}_K=\inf_{\bm w}(\cdots)$.

We prove next that {\em strong duality} holds as well:

\begin{proof} We start by proving Theorem~\ref{th:u:eq:l:app} item (2)
  Let $|\Sigma|=s$ and let $\bm A$ be the $s \times 2^n$ matrix that
  maps $\bm h$ to the vector
  $\bm A \cdot \bm h = (h(\tau))_{\tau \in \Sigma}\in \R^s$.  Let
  $\bm c \in \R^{2^n}$ be the vector $c_{\bm X} = 1$, $c_{\bm U}=0$
  for $\bm U \neq \bm X$.  The two bounds are the optimal solutions to
  the following pair of primal/dual linear programs:
  \begin{align*}
    &
      \begin{array}{ll|ll}
        \multicolumn{2}{l|}{\text{Log-L-Bound}_K}& \multicolumn{2}{|l}{\text{Log-U-Bound}_K} \\ \hline
        \multicolumn{2}{l|}{\mbox{Maximize } \bm c^T \cdot \bm h} & \multicolumn{2}{|l}{\mbox{Minimize }  \bm w^T\cdot\bm b}\\
        \mbox{where} & \bm A \cdot \bm h \leq \bm b &  \mbox{where} & \bm w^T \cdot \bm A - \bm c^T \geq \bm u^T \cdot\bm M\\
                              & -\bm M \cdot \bm h \leq 0 & &
      \end{array}
  \end{align*}
  where the primal variables are $\bm h \geq 0$, and the dual
  variables are $\bm w, \bm u \geq 0$; the reader may check that the
  two programs above form indeed a primal/dual pair.
  $\text{Log-L-Bound}_K$ is by definition the optimal value of the
  program above.  We prove that $\text{Log-U-Bound}_K$ is the value of
  the dual.  First, observe that the
  $\Sigma$-inequality~\eqref{eq:ii:lp:revisited} is equivalent to
  $(\bm w^T \cdot \bm A - \bm c^T)\cdot h \geq 0$.  We claim that this
  inequality holds $\forall \bm h \in K$ iff there exists $\bm u$
  s.t. $(\bm w,\bm u)$ is a feasible solution to the dual.  For that
  consider the following primal/dual programs with variables
  $\bm h\geq 0$ and $\bm u\geq 0$ respectively:
  \begin{align*}
    &
      \begin{array}{ll|ll}
        \multicolumn{2}{l|}{\mbox{Minimize } (\bm w^T \cdot\bm A - \bm c^T) \cdot \bm h} & \multicolumn{2}{|l}{\mbox{Maximize }  0}\\
        \mbox{where} & \bm M  \cdot \bm h \geq 0 &  \mbox{where} & \bm u^T \cdot \bm M\leq\bm w^T \bm A - \bm c^T 
      \end{array}
  \end{align*}
  The primal (left) has optimal value 0 iff the inequality
  $(\bm w^T \cdot A - \bm c^T) \cdot \bm h\geq 0$ holds forall
  $\bm h \in K$; otherwise its optimal is $-\infty$.  The dual (right)
  has optimal value 0 iff there exists a feasible solution $\bm u$;
  otherwise its optimal is $-\infty$.  Strong duality proves our
  claim.
\end{proof}

\begin{proof}
  We prove Theorem~\ref{th:u:eq:l:app} (1).  We will assume that all
  log-statistics are $b_\tau > 0$, and defer to the full paper the
  treatment of zero values of the log-statistics.

  We will use the following definition from~\cite[Example
  5.12]{boyd_vandenberghe_2004}:
  
  \begin{defn} \label{def:primal:dual:cone:program} Let $K$ be a
    proper cone (meaning: closed, convex, with a non-empty interior,
    and pointed i.e. $\bm x, - \bm x \in K$ implies $\bm x =0$).  A
    primal/dual cone program in standard form\footnote{We changed to
      the original formulation~\cite{boyd_vandenberghe_2004} by
      replacing $\bm c$ with $-\bm c$, replacing $\bm y$ with
      $-\bm y$.} is the following:
    \begin{align*}
      &
        \begin{array}{ll|ll}
          Primal && Dual& \\ \hline
          \mbox{Maximize} & \bm c^T \cdot \bm x & \mbox{Minimize} & \bm y^T\cdot b\\
          \mbox{where} & \bm A \cdot \bm x = \bm b &  \mbox{where} & (\bm y^T\cdot \bm A - \bm c^T)^T \in K^*\\
                 & \bm x \in K & & 
        \end{array}
    \end{align*}
  \end{defn}

  Denote by $P^*, D^*$ the optimal value of the primal and dual
  respectively.  {\em Weak duality} states that $P^* \leq D^*$, and is
  easy to prove.  When {\em Slater's condition} holds, which says that
  there exists $\bm x$ in the interior of $K$ such that
  $\bm A \bm x = \bm b$, then {\em strong duality} holds too:
  $P^* = D^*$.

  $\text{Log-L-Bound}_K(\Sigma,\bm b)$ and
  $\text{Log-U-Bound}_K(\Sigma,\bm b)$ can be expressed as a cone
  program, by letting $\bm A$ and $\bm c$ be the matrix and vector
  defined in the proof of Theorem~\ref{th:u:eq:l:app} (2) (thus,
  $\bm A \cdot \bm h = (h(\sigma))_{\sigma \in \Sigma}$ and
  $c_{\bm X} = 1$, $c_{\bm U}=0$ for $\bm U \neq \bm X$):
  \begin{align}
    & \begin{array}{ll|ll}
        \multicolumn{2}{l|}{\text{Log-L-Bound}_K}& \multicolumn{2}{|l}{\text{Log-U-Bound}_K} \\ \hline
        \multicolumn{2}{l|}{\mbox{Maximize } \bm c^T \cdot h} & \multicolumn{2}{|l}{\mbox{Minimize }  \bm w^T\cdot\bm b}\\
        \mbox{where} & \bm A \cdot \bm h + \bm \beta = \bm b &  \mbox{where} & (\bm w^T\cdot \bm A - \bm c^T)^T \in K^*\\
                                                 & (\bm h,\bm \beta) \in K\times \R_+^s & & \bm w\geq 0 
      \end{array} \label{eq:primal:dual:cone:h}
  \end{align}
  Here $\bm \beta$ are slack variables that convert an inequality
  $h(\sigma) \leq b_{\sigma}$ into an equality
  $h(\sigma)+ \beta_{\sigma} = b_{\sigma}$.  We leave it to the reader
  to check that these two programs are indeed primal/dual as in
  Def.~\ref{def:primal:dual:cone:program}.

  Every set $K$ s.t. $N_n \subseteq K \subseteq \Gamma_n$ has an empty
  interior, because $h(\emptyset)=$.  To circumvent that, we remove
  the $\emptyset$-dimension. Define
  $L_0 = 2^{\bm X}- \set{\emptyset}$, define the cone:
  \begin{align}
    K_0 \defeq & \Pi_{L_0}(K \cap F) \label{eq:the:cone:k}
  \end{align}
  Thus, $K_0$ removes the $\emptyset$-dimension.  The cone program
  above can be rewritten with minor changes to remove any reference to
  the dimension $\emptyset$.

  Thus, we represent both bounds, $\text{Log-L-Bound}_{K_0}$ and\linebreak
  $\text{Log-U-Bound}_{K_0}$, as the solutions to the primal/dual cone
  program~\eqref{eq:primal:dual:cone:h} over the cone $K_0$.

  It remains to check Slater's condition.  Consider the $2^n-1$ step
  functions $h^{\bm V}$ in Eq.~\eqref{eq:step:function}.  Once can
  check that they are linearly independent vectors in $\Rp^{2^n-1}$.
  Choose $2^n-1$ numbers $\varepsilon_{\bm V}>0$ small enough, such
  that, defining
  $\bm h \defeq \sum_{\bm V} \varepsilon_{\bm V} \bm h^{\bm V}$, it
  holds that $\bm h(\tau) < b_\tau$ for all $\tau \in \Sigma$: this is
  possible by choosing the coefficients $\varepsilon_{\bm V}$ small
  enough, since $b_\tau>0$.  Since $\bm h \in N_n \subseteq K$, it
  represents a feasible solution of the cone program, and it is in the
  interior of $K$, because changing the coefficients
  $\varepsilon_{\bm V}$ independently, by a small amount, continues to
  keep $\bm h$ in $N_n$ and, thus, in $K$.
\end{proof}

\subsection{Tightntess of the bounds}

\label{app:tight:not:tight}

Fix a set of log-statistics $(\Sigma, \bm b)$.  If $k\in \N$, then we
call a {\em $k$-amplification} the set of log-statistics
$(\Sigma, k\bm b)$.  Notice that these correspond to the set of
statistics $(\Sigma, \bm B^k)$.  We prove:

\begin{thm} \label{th:ae:tight} (1) For any query $Q$ and any set of
  log-statistics $(\Sigma, \bm b)$, if $Q$ guards $\Sigma$ then the
  following holds:
  \begin{align}
    \sup_k \frac{\sup_{\bm D: \bm D \models \bm B^k} \log|Q(\bm  D)|}{\text{Log-L-Bound}_{\bar \Gamma_n^*}(\Sigma,k\bm b)} = 1 \label{eq:asymptotic:tight}
  \end{align}

  (2) There exists an $\alpha$-acyclic query $Q$ with 4 variables and
  set of log-statistics $(\Sigma, \bm b)$ such that:

  \begin{align}
   \forall k\geq 1:\ \ \frac{\sup_{\bm D: \bm D \models \bm B^k} \log|Q(\bm  D)|}{\text{Log-L-Bound}_{\Gamma_n}(\Sigma,k\bm b)} \leq \frac{35}{36} \label{eq:asymptotic:polymatroid:not-tight}
  \end{align}

\end{thm}

The proof of the theorem makes essential use of the definition of the
lower bound~\eqref{eq:l:bound}: this is one reason why we introduced
it.  Item (1) states that the almost-entropic bound is tight in an
asymptotic sense.  The proof is based on Chan and Yeung's
Group-Characterization theorem~\cite{DBLP:journals/tit/ChanY02}, uses
similar ideas to those
in~\cite{DBLP:conf/icdt/GogaczT17,DBLP:conf/pods/Khamis0S17}. Item (2)
in Theorem~\ref{th:ae:tight} states that the polymatroid bound is not
tight in general: there exists a concrete query and concrete
statistics where the polymatroid bound is strictly larger than the
query cardinality.  Notice that
inequality~\eqref{eq:asymptotic:polymatroid:not-tight} is stated in
terms of logarithms: the gap between the actual query output size and
the upper polymatroid bound is an exponent, not a constant factor.

The fraction in~\eqref{eq:asymptotic:tight} is $\leq 1$, because
of weak duality~\eqref{eq:l:bound:leq:u:bound} and Theorem~\ref{th:main:bound}.  We
prove that it is $\geq 1$, following ideas from~\cite{DBLP:conf/icdt/GogaczT17,DBLP:conf/pods/Khamis0S17}

Define
$L \defeq \text{Log-L-Bound}_{\bar \Gamma_n^*}(\Sigma,k\bm b)$, and
let $\bm h^* \in \bar \Gamma_n^*$ be an optimal solution to the bound
above, i.e.  $\bm h^*\models (\Sigma, \bm b)$ and $h^*(\bm X) = L$.
We will assume that $L<\infty$, and in that case such a solution
exists, because $\bm h$ can be restricted to the compact set
$\bar \Gamma_n^* \cap \setof{\bm h}{\lp{\bm h}_\infty \leq L}$; when
$L=\infty$ then one can choose $\bm h^*$ s.t. $h^*(\bm X)$ arbitrarily
large and make minor adjustments to the argument in the rest of the
proof, which we omit.

Since $\bar \Gamma_n^*$ is the topological closure of $\Gamma_n^*$,
for all $\varepsilon > 0$ there exists $\bm h \in \Gamma_n^*$
s.t.$\bm h$ and $\bm h^*$ are $\varepsilon$-close, more precisely
there exists $\bm h \in \Gamma_n^*$ such that
$h(\bm X) \geq (1-\varepsilon) h^*(\bm X)$ and
$h(\tau) \leq (1+\varepsilon) h^*(\tau)$ for all $\tau \in \Sigma$.
Notice that $\bm h$ may slightly violate the constraints $\bm b$, more
precisely, the following hold:
\begin{align*}
  \bm h \models & (\Sigma, (1+\varepsilon)\bm b)& h(\bm X)\geq  &(1-\varepsilon)L
\end{align*}

Define $\bm h' \defeq (k-1)\bm h$, where $k\in \N$ is a large number
to be defined shortly.  We notice that $\bm h'$ is still an entropic
vector, because $\Gamma_n^*$ is closed under addition.  (It is not
closed under multiplication with non-integer constants.)  Write
$\bm h' = k (1-\frac{1}{k})\bm h = k (1-\delta)\bm h$ and observe that
$\bm h'$ is almost a $k$ amplification of $\bm h$.  We choose $k$ such
that $\varepsilon \leq \delta \leq 2\varepsilon$, and observe that,
for $\tau \in \Sigma$,
$\bm h'(\tau) = k (1-\delta)\bm h(\tau) \leq
k(1-\varepsilon)h(\tau)\leq k(1-\varepsilon^2)b_\tau$.  It follows
that the following hold:
\begin{align*}
  \bm h' \models & (\Sigma, (1-\varepsilon^2)k\bm b), & h'(\bm X) \geq & (1-2\varepsilon)^2kL
\end{align*}

At this point we would like to convert the probability space
associated to the entropic vector $\bm h'$ into a database.  However,
we cannot simply take its support and view it as a database, because
the probability distribution is non-uniform, hence $\log Q(\bm D)$
will not be equal to $h'(\bm X)$.  Instead, we use an elegant result
by Chan and Yeung~\cite{DBLP:journals/tit/ChanY02}.  (The same
argument was used in prior
work~\cite{DBLP:conf/icdt/GogaczT17,DBLP:conf/pods/Khamis0S17}, hence
only sketch the main idea here.)

Given a finite group $G$ and a subgroup $G_1 \subseteq G$, a {\em left
  coset} is a set of the form $a G_1$, for some $a \in G$.  By
Lagrange's theorem, the set of left cosets, denoted $G/G_1$, forms a
partition of $G$, and $|G/G_1| = |G|/|G_1|$.  Fix $n$ subgroups
$G_1, \ldots, G_n$, and consider the relational instance:
\begin{align}
R = & \setof{(aG_1,\ldots, aG_n)}{a \in G} \label{eq:group:realization}
\end{align}
whose set of attributes we identify, as usual, with
$\bm X=\set{X_1, \ldots, X_n}$.  Notice that
$|R| = |G|/|\bigcap_{i=1,n}G_i|$.  The entropic vector $\bm h$
associated to the relation $R$ is called a {\em group realizable
  entropic vector}, and the set of group realizable entropic vectors
is denoted by $\Upsilon_n \subseteq \Gamma_n^*$.  One can check that,
for any subset of variables $\bm U \subseteq \bm X$,
$h(\bm U) =\log |G|/|\bigcap_{X_i \in \bm U} G_i|$.  The following was
proven in~\cite{DBLP:journals/tit/ChanY02}:

\begin{thm} \label{th:chan:groups}
  For any $\bm h \in \Gamma_n^*$ there exists a sequence
  $\bm h^{(r)} \in \Upsilon_n$, such that
  $\lim_{r \rightarrow \infty} \frac{1}{r} \bm h^{(r)} = \bm h$.
\end{thm}

We complete the proof by approximating $\bm h'$ by some group
realizable entropic vector $\frac{1}{r}\bm h^{(r)}$.  Since $\bm h'$
satisfies all $k$-amplified constraints $(\Sigma, \bm b)$ with some
slack, we can choose $r$ large enough to ensure that
$\frac{1}{r}\bm h^{(r)}$ will still satisfy the $k$-amplified
constraints $(\Sigma,\bm b)$, and, similarly, that
$\frac{1}{r}h^{(r)}(\bm X) \geq (1-\varepsilon)h'(\bm X)$.  Thus, we
have:
\begin{align*}
  \bm h^{(r)} \models & (\Sigma, kr\bm b), & h^{(r)}(\bm X) \geq & (1-2\varepsilon)^3krL
\end{align*}
In other words, $\bm h^{(r)}$ satisfies the $kr$-amplified statistics,
and $h^{(r)}$ is arbitrarily close to $krL$.  Consider now the
relation~\eqref{eq:group:realization} that defines $\bm h^{(r)}$.  $R$
is totally uniform, in the sense that for any set of variables
$\bm U, \bm V$ and any two values $\bm u, \bm u'$ of $\bm U$,
$\degree_R(\bm V|\bm U=\bm u)=\degree_R(\bm V|\bm U=\bm u')$.  This implies:
\begin{align*}
  h(\bm U) = & \log|\Pi_{\bm U}(R)|, \\
  h(\bm V|\bm U) = &\log \lp{\degree_R(\bm V|\bm U)}_\infty,\\
  h(\bm U\bm V) + & (p-1) h(\bm V|\bm U) = \log\lp{\degree_R(\bm V|\bm U)}_p^p
\end{align*}
Therefore, we define the database $\bm D = (R_1^D, \ldots, R_m^D)$ by
setting $R_j^D \defeq \Pi_{\bm Y_j}(R)$.  We observe that
$\bm D \models \bm B^{kr}$ and $Q(\bm D) = R$, which implies
$\log |Q(\bm D)| = h^{(r)}(\bm X) \geq (1-2\varepsilon)^3krL$.  In
other words, $\log |Q(\bm D)|$ is arbitrarily close to the
$kr$-amplification of the bound $L$.  Since $\varepsilon>0$ was
arbitrary, it follows that the fraction in~\eqref{eq:asymptotic:tight}
is $\geq 1$, as required.

Next, we prove item (2) of Theorem~\ref{th:ae:tight}.

We first describe a non-Shannon inequality, which we later use to
derive a query and statistics for which the polymatroid bound is not
tight.

\begin{prop}\label{prop:nonShannon-lp}
    The following is a non-Shannon inequality:
\begin{align}
9h(ABXY) \leq \ & [h(ABXY) + 4h(B|AXY)] + [h(ABXY) + h(A|BXY)] + \nonumber\\ 
 & [h(ABXY) + h(XY|AB)] + h(BX) + h(BY) + \nonumber\\ 
 & \frac{1}{2}[h(XY) + 2h(Y|X) + h(XY) + 2h(X|Y)] + \nonumber\\ 
 &\frac{1}{2}[h(AY) + 2h(Y|A) + h(AY) + 2h(A|Y)] + \nonumber\\
 & [h(AX) + h(A|X)] + h(AX). \label{ineq:nonShannon-lp}
\end{align}
\end{prop}
\begin{proof}
    We prove Inequality~\eqref{ineq:nonShannon-lp} by a series of transformations from Zhang and Yeung's non-Shannon inequality~\cite{ZhangY98}:
\begin{align}
    I(X;Y) \leq & 2I(X;Y |A) + I(X;Y|B) + I(A;B) + I(A;Y|X) + I(A;X | Y) \label{eq:Zhang-Yeung}
\end{align}

Recall that $I(X;Y | A)$ represents the mutual information of the variables $X$ and $Y$ given variable $A$ and is defined using the entropy function $h$ over these variables as follows:
\begin{align*}
    I(X;Y|A) = h(AX) + h(AY) - h(AXY) - h(A) 
\end{align*}

We expand Equation~\eqref{eq:Zhang-Yeung} to use the entropy function only:

\begin{align*}
    0 \leq \ 
    & -(h(X) + h(Y) - h(XY)) + 2 h(AX) + \\
    & 2 h(AY) - 2h(AXY) - 2h(A) + \\
    & h(BX) + h(BY) - h(BXY) - h(B) + h(A) + h(B) - h(AB) + \\
    & h(AX) + h(XY) - h(AXY) - h(X) +  \\
    & h(AY) + h(XY) -h(AXY) - h(Y) \\
    & \Leftrightarrow \\
 0 \leq \ & 3 h(XY) - 2 h(X) - 2h(Y) - 4h(AXY) - h(BXY) + \\ 
          & 3 h(AX) + 3h(AY) + h(BX) + h(BY) - h(AB) - h(A)   
\end{align*}

We add $9h(ABXY)$ to both sides of the inequality and use the following equalities:
\begin{align*}
    h(ABXY) - h(AXY) &= h(B|AXY) \\
    h(ABXY) - h(BXY) &= h(A|BXY) \\
    h(ABXY) - h(AB)  &= h(XY|AB) \\
    h(XY) - h(X)     &= h(Y|X) \\
    h(XY) - h(Y)     &= h(X|Y) \\
    h(AY) - h(A)     &= h(Y|A) \\
    h(AY) - h(Y)     &= h(A|Y) \\
    h(AX) - h(X)     &= h(A|X)
\end{align*}

The inequality becomes:

\begin{align*}
 9h(ABXY) \leq \ & [h(ABXY) + 4h(B|AXY)] + [h(ABXY) + h(A|BXY)] + \\ 
 & [h(ABXY) + h(XY|AB)] + h(BX) + h(BY) + \\
 & [h(XY) + h(Y|X) + h(X|Y)] + \\
 & [h(AY) + h(Y|A) + h(A|Y)] + \\
 & [h(AX) + h(A|X)] + h(AX) \\
 &\Leftrightarrow \\
9h(ABXY) \leq \ & [h(ABXY) + 4h(B|AXY)] + [h(ABXY) + h(A|BXY)] + \\ 
 & [h(ABXY) + h(XY|AB)] + h(BX) + h(BY) + \\ 
 & \frac{1}{2}[h(XY) + 2h(Y|X) + h(XY) + 2h(X|Y)] + \\ 
 &\frac{1}{2}[h(AY) + 2h(Y|A) + h(AY) + 2h(A|Y)] + \\
 & [h(AX) + h(A|X)] + h(AX).
\end{align*}
\end{proof}

\paragraph{Query and its upper bound derived from non-Shannon inequality.}
Consider the following join query, derived from Inequality~\eqref{ineq:nonShannon-lp}:
\begin{align*}
    Q(A,B,X,Y) = &R_1(A,B,X,Y) \wedge R_2(B,X) \wedge R_3(B,Y), \\
                 &R_4(X,Y) \wedge R_5(A,Y) \wedge R_6(A,X)
\end{align*}
The hypergraph of $Q$ has two triangles $R_2-R_3-R_4$ and $R_4-R_5-R_6$, both included in the hyperedge $R_1$. $Q$ is thus $\alpha$-acyclic.

In Inequality~\eqref{ineq:nonShannon-lp}, some terms are grouped to show their correspondence to the $\ell_p$-norms on degree vectors. 

We also use the equality $h(\bm V\bm U) + (p-1)h(\bm V|\bm U) = ph(\bm V\bm U) - (p-1)h(\bm U) = h(\bm U) + p(h(\bm V\bm U) - h(\bm U)) = h(\bm U) + ph(\bm V|\bm U)$, which gives a different reading of Inequality~\eqref{eq:main-inequality}.

Using $h(ABXY) = \log |Q|$ and Inequality~\eqref{eq:bound:lp} with $w_i=1$, we obtain the following entropic bound on the query output size:
\begin{align*}
    |Q(ABXY)|^9 \leq \ & ||\text{deg}_{R_1}(B|AXY)||_5^5 \cdot ||\text{deg}_{R_1}(A|BXY)||_2^2 \cdot \\
    & ||\text{deg}_{R_1}(XY|AB)||_2^2 \cdot |R_2| \cdot |R_3| \cdot \\
    & ||\text{deg}_{R_4}(Y|X)||_3^{3/2} \cdot ||\text{deg}_{R_4}(X|Y)||_3^{3/2} \cdot \\
    & ||\text{deg}_{R_5}(Y|A)||_3^{3/2} \cdot ||\text{deg}_{R_5}(A|Y)||_3^{3/2} \cdot \\
    & ||\text{deg}_{R_6}(A|X)||_2^2 \cdot |R_6|
\end{align*}

\paragraph{Statistics.}
Consider the following log-statistics:
\begin{align*}
    \Sigma &= \{(
    B|AXY),(A|BXY),(XY|AB),(B,X),(B,Y),(Y|X),(X|Y),\\
    &\hspace*{1.5em}(Y|A),(A|Y),(A|X),(A,X)\}\\
    \bm b &= \{b_1= 4/5, b_2 = b_3 = b_{10} = 2, b_4 = b_5 = b_{11} = 3, \\ 
    &\hspace*{1.2em} b_6 = b_7 = b_8 = b_9 = 5/3\}, \text{ where}\\
    &\log\lp{\degree_{R_1}(B|AXY)}_5 \leq b_1, \hspace*{1em}\log\lp{\degree_{R_1}(A|BXY)}_2 \leq b_2,\\
    &\log\lp{\degree_{R_1}(XY|AB)}_2 \leq b_3, \\ 
    &\log\lp{\degree_{R_2}(BX)}_1 \leq b_4, \hspace*{1em}\log\lp{\degree_{R_3}(BY)}_1  \leq b_5, \\
    &\log\lp{\degree_{R_4}(Y|X)}_3 \leq b_6, \hspace*{1em}\log\lp{\degree_{R_4}(X|Y)}_3 \leq b_7, \\ 
    &\log\lp{\degree_{R_5}(Y|A)}_3 \leq b_8, \hspace*{1em}\log\lp{\degree_{R_5}(A|Y)}_3 \leq b_9, \\
    &\log\lp{\degree_{R_6}(A|X)}_2 \leq b_{10}, \hspace*{1em}\log\lp{\degree_{R_6}(AX)}_1 \leq b_{11}. 
\end{align*}

\begin{figure}
    \centering
    \input{yeung-zhang-h.tex}
    \caption{A lattice of closed sets and the polymatroid from ~\cite{ZhangY98} defined on the lattice.}
    \label{fig:yeung-zhang-polymatroid}
\end{figure}

These log-statistics $(\Sigma,\bm b)$ are valid for the polymatroid $\bm h$ in Figure~\ref{fig:yeung-zhang-polymatroid}, i.e., $\bm h\models (\Sigma,\bm b)$. They are constructed as follows:

\begin{align*}
    \Sigma &= \{(
    B|AXY),(A|BXY),(XY|AB),(B,X),(B,Y),(Y|X),(X|Y),\\
    &\hspace*{1.5em}(Y|A),(A|Y),(A|X),(A,X)\}\\
    \bm b &= \{b_1= 4/5, b_2 = b_3 = b_{10} = 2, b_4 = b_5 = b_{11} = 3, \\ 
    &\hspace*{1.2em} b_6 = b_7 = b_8 = b_9 = 5/3\}, \text{ where}\\
\end{align*}
\begin{align*}
    &\log\lp{\degree_{R_1}(B|AXY)}_5 \leq b_1, \hspace*{1em}\log\lp{\degree_{R_1}(A|BXY)}_2 \leq b_2,\\
    &\log\lp{\degree_{R_1}(XY|AB)}_2 \leq b_3, \\ 
    &\log\lp{\degree_{R_2}(BX)}_1 \leq b_4, \hspace*{1em}\log\lp{\degree_{R_3}(BY)}_1  \leq b_5, \\
    &\log\lp{\degree_{R_4}(Y|X)}_3 \leq b_6, \hspace*{1em}\log\lp{\degree_{R_4}(X|Y)}_3 \leq b_7, \\ 
    &\log\lp{\degree_{R_5}(Y|A)}_3 \leq b_8, \hspace*{1em}\log\lp{\degree_{R_5}(A|Y)}_3 \leq b_9, \\
    &\log\lp{\degree_{R_6}(A|X)}_2 \leq b_{10}, \hspace*{1em}\log\lp{\degree_{R_6}(AX)}_1 \leq b_{11}. 
\end{align*}

$\Sigma$ is constructed immediately from Inequality~\eqref{ineq:nonShannon-lp}: There is one element of $\Sigma$ for each entropic term in the information inequality. This is valid for $\bm h$, as all terms only use a subset of the variables $A,B,X,Y$.

The construction of the log-statistics $\bm b$ is slightly more involved. Figure~\ref{fig:yeung-zhang-polymatroid} depicts a lattice of closed sets. It does not show elements whose entropy is the same as for an element that includes them: For instance, $AXY$ is not shown and its entropy is the same as of the element that contains $AXY$ and is shown in the lattice, i.e., $h(AXY) = h(XYAB) = 4$. Similarly, $h(BXY) = h(AB) = h(XYAB) = 4$.

The log-statistics can then be derived as follows:

\begin{align*}
    \log\lp{\degree_{R_1}(B|AXY)}_5^5 &= h(ABXY) + 4h(B|AXY) \\ 
    &= 5h(ABXY) - 4h(AXY) = 5\cdot 4 - 4\cdot 4 = 4 \\
    &\Rightarrow \log\lp{\degree_{R_1}(B|AXY)}_5 = 4/5 = b_1\\
    \log\lp{\degree_{R_1}(A|BXY)}_2^2 &= 2\cdot h(ABXY) - h(BXY) = 4\\
    &\Rightarrow \log\lp{\degree_{R_1}(A|BXY)}_2 = 2 = b_2 \\
    \log\lp{\degree_{R_1}(XY|AB)}_2^2 &= 2h(ABXY) - h(AB) = 4\\ 
    &\Rightarrow  \log\lp{\degree_{R_1}(XY|AB)}_2 = 2 = b_3 \\
    \log\lp{\degree_{R_2}(BX)}_1 &= h(BX) = 3 = b_4 \\
    \log\lp{\degree_{R_3}(BY)}_1 &= h(BY) = 3 = b_5 \\
    \log\lp{\degree_{R_4}(Y|X)}_3^3 &= h(XY) + 2h(Y|X) \\
    &= 3h(XY) - 2h(X) = 3\cdot 3 - 2\cdot 2 = 5\\
    &\Rightarrow  \log\lp{\degree_{R_4}(Y|X)}_3 = 5/3 = b_6 \\
    \log\lp{\degree_{R_4}(X|Y)}_3^3 &= h(XY) + 2h(X|Y) \\
    &= 3h(XY) - 2h(Y) = 3\cdot 3 - 2\cdot 2 = 5\\
    &\Rightarrow  \log\lp{\degree_{R_4}(X|Y)}_3 = 5/3 = b_7 \\
    \log\lp{\degree_{R_5}(Y|A)}_3^3 &= h(AY) + 2h(Y|A) \\
    &= 3h(AY) - 2h(A) = 3\cdot 3 - 2\cdot 2 = 5\\
    &\Rightarrow  \log\lp{\degree_{R_5}(Y|A)}_3 = 5/3 = b_8 \\
\end{align*}
\begin{align*}
    \log\lp{\degree_{R_5}(A|Y)}_3^3 &= h(AY) + 2h(A|Y) \\
    &= 3h(AY) - 2h(Y) = 3\cdot 3 - 2\cdot 2 = 5\\
    &\Rightarrow  \log\lp{\degree_{R_5}(A|Y)}_3 = 5/3 = b_9 \\
    \log\lp{\degree_{R_6}(A|X)}_2^2 &= h(AX) + h(A|X) \\
    &= 2h(AX) - h(X) = 2\cdot 3 - 2 = 4 \\
    &\Rightarrow  \log\lp{\degree_{R_6}(A|X)}_2 = 2 = b_{10} \\
    \log\lp{\degree_{R_6}(AX)}_1 &= h(AX) = 3 = b_{11}.
\end{align*}

\paragraph{Concluding the argument that the bound is not tight for the polymatroid.}

If we consider the scaled log-statistics $k\bm b$ for any scale factor $k>0$, Definition~\ref{eq:u:bound} yields:
\begin{align*}
    &\text{Log-U-Bound}_{\Gamma_n^*}(\Sigma,k\bm b) \leq k\cdot \\
    &\frac{5b_1+2(b_2+b_3+b_{10}) + b_4 + b_5 + b_{11} + \frac{3}{2}(b_6+b_7+b_8+b_9)}{9} 
     = \frac{35k}{9}
\end{align*}
Therefore, for any database $\bm D$ for which $\bm D\models (\Sigma,\bm B^k)$ it holds that $\log |Q(\bm D)| \leq \frac{35k}{9}$. 

On the other hand, consider a polymatroid $k\bm h$, where $\bm h$ is the polymatroid in Figure~\ref{fig:yeung-zhang-polymatroid}. Since $h(ABXY)=4$, it follows that $k\cdot h(ABXY)=4k$. Since $\bm h\models (\Sigma,\bm b)$, it follows that $k\bm h\models (\Sigma, k\bm b)$. Therefore, 
  \begin{align*}  
    \frac{\sup_{\bm D: \bm D \models \bm B^k} \log|Q(\bm  D)|}{\text{Log-L-Bound}_{\Gamma_n}(\Sigma,k\bm b)} \leq \frac{35k}{9}\cdot\frac{1}{4k} = \frac{35}{36}.
  \end{align*}

This proves Inequality~\eqref{eq:asymptotic:polymatroid:not-tight} in Item (2) of Theorem~\ref{th:ae:tight}.


%% file: yeung-zhang-h.tex
\colorlet{DarkGreen}{green!60!black}
\begin{tikzpicture}[scale=0.8]
	\tikzstyle{node}=[]
	\tikzstyle{label}=[DarkGreen]
	\tikzstyle{path}=[gray]
	\node[node] at (4, 0) (XYAB){$XYAB$};    \node[label, below right=-0.3 of XYAB] {$h=4$};

	\node[node] at (0, -2) (AX){$AX$};       \node[label, below right=-0.3 of AX] {$h=3$};
	\node[node] at (2, -2) (AY){$AY$};       \node[label, below right=-0.3 of AY] {$h=3$};
	\node[node] at (4, -2) (XY){$XY$};       \node[label, below right=-0.3 of XY] {$h=3$};
	\node[node] at (6, -2) (XB){$XB$};       \node[label, below right=-0.3 of XB] {$h=3$};
	\node[node] at (8, -2) (YB){$YB$};       \node[label, below right=-0.3 of YB] {$h=3$};

	\node[node] at (.5, -4) (A){$A$};        \node[label, below right=-0.3 of A] {$h=2$};
	\node[node] at (.5+7/3, -4) (X){$X$};    \node[label, below right=-0.3 of X] {$h=2$};
	\node[node] at (.5+14/3, -4) (Y){$Y$};   \node[label, below right=-0.3 of Y] {$h=2$};
	\node[node] at (7.5, -4) (B){$B$};       \node[label, below right=-0.3 of B] {$h=2$};

	\node[node] at (4, -6) (0){$\emptyset$}; \node[label, below right=-0.3 of 0] {$h=0$};

	\path[path] (0) edge (X);
	\path[path] (0) edge (A);
	\path[path] (0) edge (B);
	\path[path] (0) edge (Y);

	\path[path] (X) edge (AX);
	\path[path] (X) edge (XB);
	\path[path] (X) edge (XY);
	\path[path] (Y) edge (YB);
	\path[path] (Y) edge (AY);
	\path[path] (Y) edge (XY);
	\path[path] (A) edge (AX);
	\path[path] (A) edge (AY);
	\path[path] (B) edge (XB);
	\path[path] (B) edge (YB);

	\path[path] (AX) edge (XYAB);
	\path[path] (XB) edge (XYAB);
	\path[path] (XY) edge (XYAB);
	\path[path] (AY) edge (XYAB);
	\path[path] (YB) edge (XYAB);
\end{tikzpicture}

%% file: appendix-6.tex

\section{Proof from Sec.~\ref{sec:simple:inequalities}}

\label{app:proof:th:simple}

We start by proving Theorem~\ref{th:simple}.  The proof follows
immediately from a result in~\cite{DBLP:journals/tods/KhamisKNS21}:

\begin{thm} \label{th:simple:inequalities} Let $\Sigma$ be a simple
  set of LP-statistics.  Consider the
  $\Sigma$-inequality~\eqref{eq:ii:lp:revisited}.  Then the following
  are equivalent:
  \begin{itemize}
  \item Eq.~\eqref{eq:ii:lp:revisited} is valid for all $\bm h \in \Gamma_n$.
  \item Eq.~\eqref{eq:ii:lp:revisited} is valid for all $\bm h \in \bar \Gamma_n^*$.
  \item Eq.~\eqref{eq:ii:lp:revisited} is valid for all $\bm h \in N_n$.
  \end{itemize}
\end{thm}

In other words, if the inequality~\eqref{eq:ii:lp:revisited} is
simple, then it is valid for all polymatroids iff it is valid for all
(almost-) entropic vectors, iff it is valid for all normal
polymatroids.  This immediately implies~\ref{th:simple}.

Finally, we prove that inequality~\eqref{eq:ex:normal:database} is a
Shannon inequality.  This follows by observing that it is the sum of
the following  basic Shannon inequalities:

\begin{align*}
  2h(X) + 2h(Y|X) + 2h(Z|Y) \geq & 2h(XYZ) \\
  2h(Y) + 2h(Z|Y) + 2h(X|Z) \geq & 2h(XYZ) \\
  2h(Z) + 2h(X|Z) + 2h(Y|X) \geq & 2h(XYZ)
\end{align*}